\documentclass[preprint,sort&compress]{elsarticle}
\usepackage[left=2.0cm,right=2.0cm,top=1cm,bottom=2cm,bindingoffset=0cm]{geometry}
\usepackage{float,eepic,subfig}
\usepackage{amsmath,amssymb,amsthm}
\usepackage{enumitem}
\usepackage{eepic}
\usepackage{amsmath,amssymb,amsthm}
\usepackage{multicol}
\usepackage{array}

\newtheorem{theorem}{Theorem}
\newtheorem{proposition}{Proposition}
\newtheorem{corollary}{Corollary}

\newtheorem{conseq}{Corollary}

\DeclareMathOperator{\LEs}{LE}
\DeclareMathOperator{\LCEs}{LCE}
\DeclareMathOperator{\Tr}{Tr}
 \journal{arXiv} 

\begin{document}
\begin{frontmatter}

\title{The Lyapunov dimension and its computation
for self-excited and hidden attractors in the Glukhovsky-Dolzhansky fluid convection model}

\author[fin,spb]{Kuznetsov N.~V.\corref{cor}}
\author[spb,ipmash]{Leonov G.~A.}
\author[spb]{Mokaev T.~N.}
\cortext[cor]{Corresponding author email: nkuznetsov239@gmail.com}

\address[fin]{
Dept. of Mathematical Information Technology,\
University of Jyv\"{a}skyl\"{a}, Finland}
\address[spb]{
Faculty of Mathematics and Mechanics,\
Saint-Petersburg State University, Russia}
\address[ipmash]{Institute of Problems of Mechanical Engineering RAS, Russia}

\begin{abstract}
Consideration of various hydrodynamic phenomena
involves the study of the Navier-Stokes (N-S) equations,
what is hard enough for analytical and numerical investigations
since already in three-dimensional (3D) case it is a challenging task to study
the limit behavior of N-S solutions.
The low-order models (LOMs) derived from the initial N-S equations
by Galerkin method allow one to overcome difficulties in studying the
limit behavior and existence of attractors.
Among the simple LOMs with chaotic attractors
there are famous Lorenz system, which is an approximate model of two-dimensional convective flow
and Glukhovsky-Dolzhansky model,
which describes a convective process
in three-dimensional rotating fluid and
can be considered as an approximate model of the World Ocean.
One of the widely used dimensional characteristics of attractors is the Lyapunov dimension.
In the study we follow a rigorous approach for the definition
of the Lyapunov dimension and justification of its computation by the Kaplan-Yorke formula,
without using statistical physics assumptions.
The exact Lyapunov dimension formula for the global attractors
is obtained and peculiarities
of the Lyapunov dimension estimation for self-excited and hidden attractors
are discussed.
A tutorial on numerical estimation of the Lyapunov dimension
on the example of the Glukhovsky-Dolzhansky model is presented.

\end{abstract}


\end{frontmatter}

\section{Introduction}
The main difficulties in studying fluid motion
are related to infinite number of degrees of freedom of hydrodynamic objects.
To overcome these difficulties, one may use an approximation
(e.g., applying Galerkin method \citep{Thompson-1961}) 
of system of equations, describing the considered object with an infinite number of degrees of freedom,
by a system of equations with a finite number of degrees of freedom.
Resulting finite-dimensional analogues of the hydrodynamic equations,
called low-order models,
turn out to be more convenient for analytical and numerical investigations
\cite{Obukhov-1973,Glukhovsky-1982,GluhovskyG-2016}. 
Among the famous physically sounded low-order models there are the Lorenz model  \cite{Lorenz-1963}
(describing the Rayleigh-B\'{e}nard convection),
the Vallis model \cite{Vallis-1986} (describing El Ni\~{n}o climate phenomenon),  
and the Glukhovsky-Dolzhansky model \cite{GlukhovskyD-1980}
(describing fluid convection inside the rotating ellipsoidal cavity under
the horizontal heating).
One of the substantial features of these models is the existence of chaotic attractors in
their phase spaces.
From both theoretical and practical perspective it is important to
localize these attractors  \cite{Tucker-1999,Stewart-2000},
study their basins of attraction  \cite{MenckHMK-2013,PisarchikF-2014,DudkowskiJKKLP-2016},
and estimate their dimensions  \cite{GrassbergerP-1983}
with respect to varying parameters.

In the present paper a there-dimensional model,
describing the convection of fluid within an
ellipsoidal rotating cavity under an external horizontal heating,
is considered.
This model was suggested by Glukhovsky and Dolghansky~\cite{GlukhovskyD-1980} (G-D)
and can be considered as an approximate model of the World Ocean.
The mathematical G-D model is described by the following system of ODEs:
\begin{equation}
\begin{array}{l}
     \dot{x}  =  -\sigma x + z + a_0 yz, \\
     \dot{y}  =  R - y - xz, \\
     \dot{z}  =  -z + xy,
\end{array} \label{sys:conv_fluid}
\end{equation}
where $\sigma$, $R$, $a_0$ are positive parameters.

After the change of variables:
\begin{equation}\label{sys:conv_fluid:change_var}
  (x, \, y, \, z) \quad \rightarrow \quad \left(x, \, R - \frac{\sigma}{a_0 R+1} z, \, \frac{\sigma}{a_0 R+1} y\right),
\end{equation}
system~\eqref{sys:conv_fluid} takes the form of generalized Lorenz system
\begin{equation}
\begin{array}{l}
  \dot{x} $ = $ -\sigma x + \sigma y - A y z,\\
  \dot{y} $ = $ r x - y - x z, \\
    \dot{z} $ = $ -b z + x y,
\end{array}
\label{sys:lorenz-general}
\end{equation}
where
\begin{equation}
  b = 1, \quad A = \frac{a_0\sigma^2}{(a_0 R + 1)^2}, \quad
  r = \frac{R}{\sigma}(a_0 R + 1). \label{sys:conv_fluid:param}
\end{equation}
If
\begin{equation}\label{sys:gen-lorenz:change_var}
  R = r (\sigma - Ar) > 0, \quad
  a_0 = \frac{A}{(\sigma - Ar)^2} > 0, \quad
  b = 1.
\end{equation}
then we have the inverse transformation
\[
  (x, \, y, \, z) \to
  \left( x, \, \frac{1}{\sigma - Ar} \, z, \, r - \frac{1}{\sigma - Ar} \, y\right).
\]
For $A = 0$ system \eqref{sys:lorenz-general} coincides with
the classical Lorenz system \cite{Lorenz-1963}.

System \eqref{sys:lorenz-general} with the parameters $r$, $\sigma$, $b > 0$
is mentioned first in the work of Rabinovich \cite{Rabinovich-1978}
and in the case $A < 0$ can be transformed \cite{LeonovB-1992}
to the Rabinovich system of waves interaction
in plasma \cite{PikovskiRT-1978,KuznetsovLMS-2016-INCAAM}.
Following Glukhovsky and Dolghansky \cite{GlukhovskyD-1980},
consider system \eqref{sys:lorenz-general} under the physically sounded assumption that
$r$, $\sigma$, $b$, $A$ are positive.

\captionsetup[subfig]{margin=10pt,format=hang}
\begin{figure}[h!]
    \centering
    \subfloat[
    Monostability ($r = 687.5$):
    trajectories from almost all initial points except for a set of zero measure
    (including unstable equilibria $S_{0,1,2}$)
    tend to the same chaotic attractor $K_{\text{self-excited}}$
    (self-excited attractor with respect to all three equilibria: e.g.
    one-dimensional unstable manifold of $S_0$ is attracted to $K_{\text{self-excited}}$
    and, thus, visualizes it).
    ] {
      \label{fig:fluid_conv:attr_se}
      \includegraphics[width=0.49\textwidth]{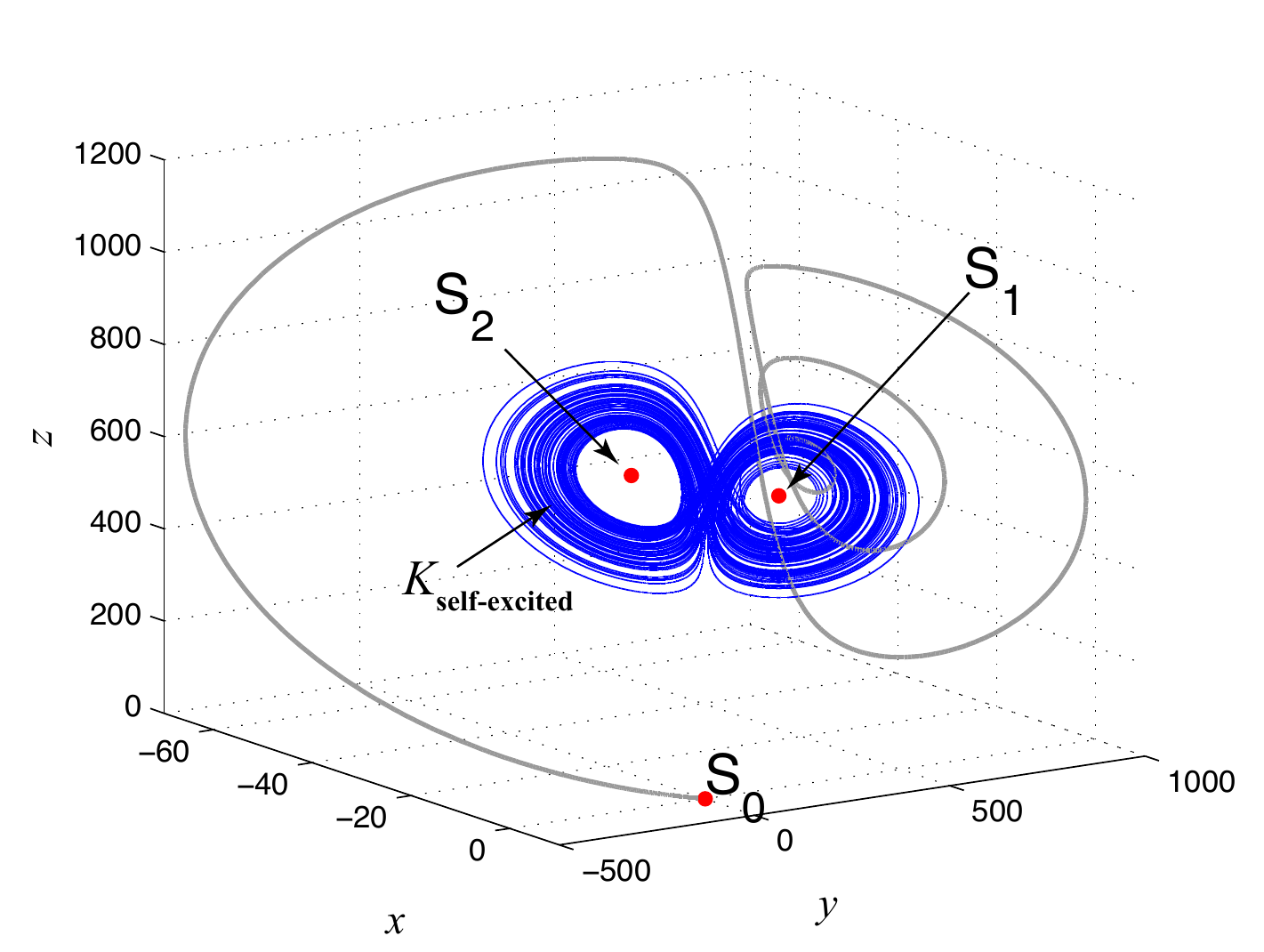}
    }~~
    \subfloat[
    Multistability ($r = 700$):
    Coexistence of three local attractors ---
    two stable equilibria $S_{1,2}$ and one chaotic attractor $K_{\text{hidden}}$
    (hidden attractor, which basin of attraction
    does not overlap with an arbitrarily small vicinity of equilibria:
    one-dimensional unstable manifold of $S_0$ is attracted to $S_{1,2}$).
    ] {
      \label{fig:fluid_conv:attr_hid}
      \includegraphics[width=0.49\textwidth]{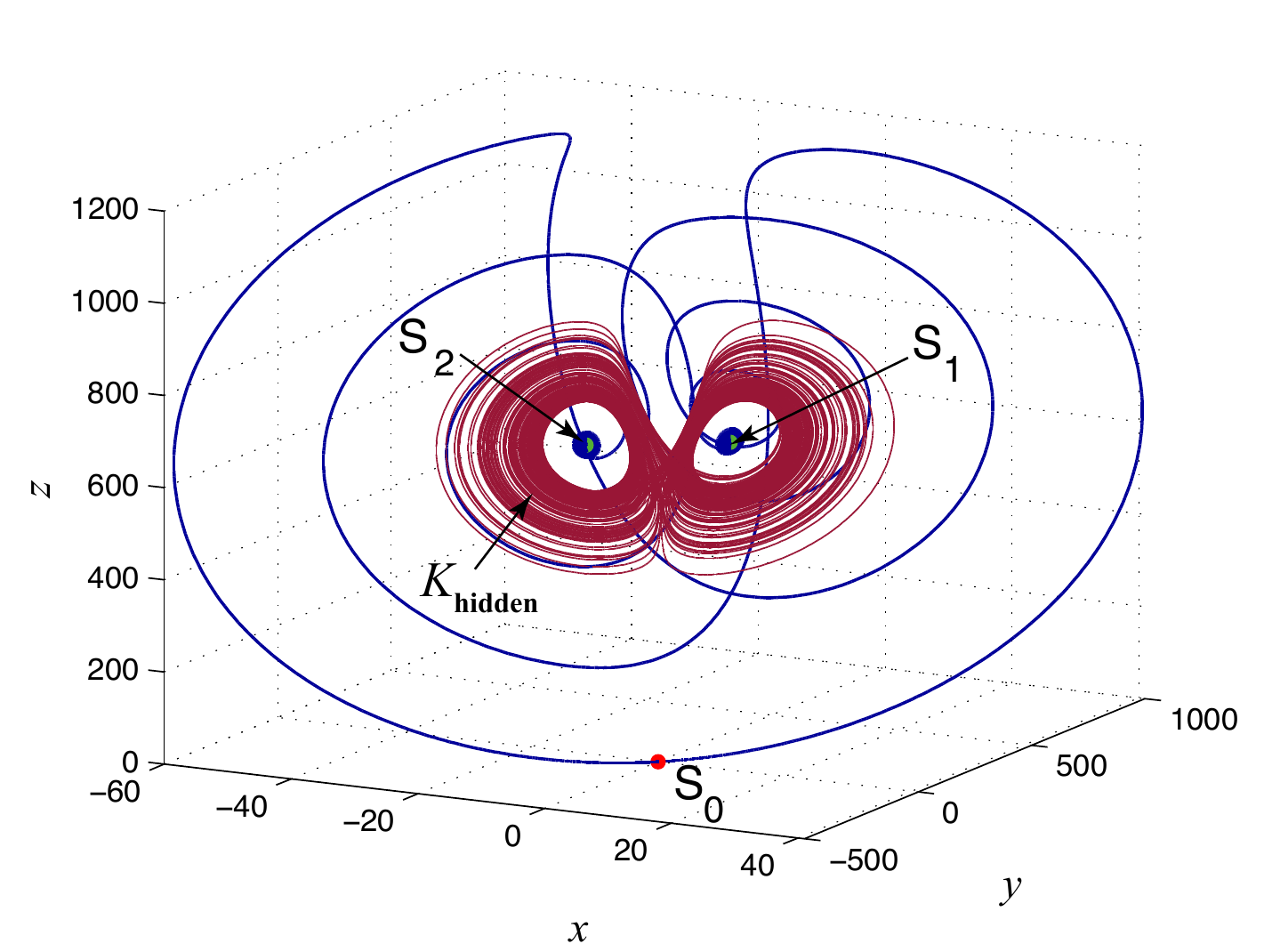}
    }
    \caption{Self-excited and hidden attractors in system \eqref{sys:lorenz-general}
    with $b = 1$, $\sigma = 4$, $A = 0.0052$.}
  \label{fig:attractors}
\end{figure}

Systems \eqref{sys:conv_fluid} and \eqref{sys:lorenz-general}
are of particular interest because they
have chaotic attractors (Fig.~\ref{fig:attractors}).
By numerical simulations
in the case when parameter $\sigma = 4$
it is obtained \cite{GlukhovskyD-1980}
certain values of the parameters
for which systems
\eqref{sys:conv_fluid} and \eqref{sys:lorenz-general}
possess self-excited attractors (Fig. \ref{fig:fluid_conv:attr_se}).
An attractor is called a \emph{self-excited attractor} if its basin of attraction
intersects an arbitrarily small open neighborhood of an equilibrium,
otherwise it is called a \emph{hidden attractor}
 \cite{LeonovKV-2011-PLA,LeonovKV-2012-PhysD,LeonovK-2013-IJBC,LeonovKM-2015-EPJST}.
Self-excited attractors are relatively simple for localization and can be obtained
using trajectories from an arbitrary small neighborhood of unstable equilibrium.
The use of the term  \emph{self-excited oscillation} or {\it self-oscillations}
can be traced back to the works of H.G.~Barkhausen and A.A.~Andronov,
where it describes the generation and maintenance of a periodic motion in mechanical
and electrical models by a source of power that lacks any corresponding periodicity
(e.g., stable limit cycle in the van der Pol oscillator) \cite{AndronovVKh-1966,Jenkins-2013}.
We use this notion for attractors of dynamical systems
to describe the existence of transient process
from a small vicinity of an unstable equilibrium
to an attractor.
If there is no such a transient process for an attractor,
it is called a hidden attractor.
The \emph{hidden and self-excited classification of attractors} was introduced by Leonov and Kuznetsov in connection with the discovery of hidden Chua attractor \cite{KuznetsovLV-2010-IFAC,LeonovKV-2011-PLA,KuznetsovKLV-2013} and
its rigorous consideration for autonomous and nonautonomous systems can be found in
 \cite{LeonovK-2013-IJBC,LeonovKM-2015-EPJST,Kuznetsov-2016,DudkowskiJKKLP-2016}.
For example, hidden attractors are attractors in systems
without equilibria or with only one stable equilibrium
(a special case of multistability and coexistence of attractors).
Some examples of hidden attractors can be found in
 \cite{ShahzadPAJH-2015-HA,BrezetskyiDK-2015-HA,JafariSN-2015-HA,ZhusubaliyevMCM-2015-HA,SahaSRC-2015-HA,Semenov20151553,FengW-2015-HA,Li20151493,FengPW-2015-HA,Sprott20151409,PhamVVJ-2015-HA,VaidyanathanPV-2015-HA,
Danca-2016-HA,Zelinka-2016-HA,DudkowskiJKKLP-2016}.
Recently hidden attractors were localized \cite{LeonovKM-2015-CNSNS,LeonovKM-2015-EPJST}
in systems \eqref{sys:conv_fluid} and \eqref{sys:lorenz-general}
(Fig. \ref{fig:fluid_conv:attr_hid}).

By the Lyapunov function $V(x,y,z) = \frac{1}{2} \left( x^2 + y^2 + (A+1)
\left(z - \frac{\sigma + r}{A + 1} \right)^2 \right)$ 
it is proved  \cite{LeonovB-1992} that system~\eqref{sys:lorenz-general}
possesses a bounded absorbing ellipsoid
(thus it is \emph{dissipative in the sense of Levinson} \cite{LeonovKM-2015-EPJST})
\begin{equation}\label{absorb_set}
  \mathcal{B}(r,\sigma,A) = \left\{(x,y,z) \in \mathbb{R}^3 ~|~ V(x,y,z) \leq \frac{(\sigma + r)^2}{2(A+1)}\right\}
\end{equation}
and, thus, has 
a global attractor and generates a dynamical system.
Also it is known \cite{LeonovB-1992} that for $b = 1$
the global attractor is located in the positive invariant set
\begin{equation}\label{localize_set}
  \Omega = \left\{y^2 + z^2 \leq 2 r z\right\}.
\end{equation}

To obtain necessary conditions of the existence of self-excited attractor,
we consider the stability of equilibria in system \eqref{sys:lorenz-general}.
According to \cite{LeonovB-1992}, we have the following:
if $r < 1$, then there is a unique equilibrium
${\bf \rm S_0} = (0, \, 0, \, 0)$ (the trivial case).
If $r > 1$, then \eqref{sys:lorenz-general} has three equilibria:
${\bf \rm S_0} = (0, \, 0, \, 0)$ -- saddle, and
${\bf \rm S_{1,2}} = (\pm x_1, \, \pm y_1, \, z_1)$, where
$$
  x_1 = \frac{\sigma b \sqrt{\xi}}{\sigma b + A \xi}, \quad
  y_1 = \sqrt{\xi}, \quad
  z_1 = \frac{\sigma \xi}{\sigma b + A \xi}, \quad
  \xi = \frac{\sigma b}{2 A^2}
  \left[ A (r-2) - \sigma + \sqrt{(Ar-\sigma)^2 + 4A\sigma} \right].
$$
The stability of $S_{1,2}$ depends on the parameters, e.g.
for $b = 1$, $\sigma = 4$ the stability domain \cite{LeonovKM-2015-EPJST}
is shown in Fig.~\ref{fig:stab_dom}.
Here for parameters from the non-shaded region, a self-excited attractor
can be localized by a trajectory from a small neighborhood of $S_0$, $S_1$ or $S_2$.
\begin{figure}[h!]
    \centering
    \includegraphics[width=0.5\textwidth]{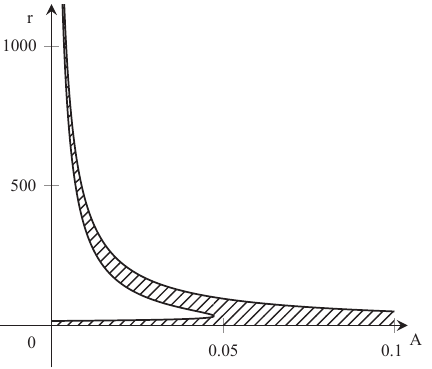}
    \caption{Stability domain (shaded) of the equilibria $S_{1,2}$ for $b = 1$, $\sigma = 4$.}
    \label{fig:stab_dom}
\end{figure}
To get a numerical characteristic of chaos in a system
using numerical methods,
it is possible to compute a local Lyapunov dimension for this trajectory,
what gives an estimation of Lyapunov dimension of the corresponding self-excited attractor.
For the hidden attractor visualization in the considered systems
we need to use special analytical-numerical procedures
of searching for a point in its domain of attraction \cite{LeonovKM-2015-CNSNS,LeonovKM-2015-EPJST}.
Thus, the estimation of Lyapunov dimension of hidden attractors in the considered systems
is a challenging task.

\section{Preliminaries. Analytical estimates of the Lyapunov dimension}\label{sec:1}

\subsection{The Lyapunov dimension and Kaplan-Yorke formula}
The concept of the Lyapunov dimension was suggested in the seminal paper
by Kaplan and Yorke \cite{KaplanY-1979} for estimating the Hausdorff dimension
of attractors.
Later it has been developed and rigorously justified in a number of papers
(see, e.g. \cite{DouadyO-1980,Ledrappier-1981,EdenFT-1991} and others).
The direct numerical computation of the Hausdorff dimension
of chaotic attractors is often a problem of high numerical complexity
(see, e.g. the discussion in \cite{RusselHO-1980}),
thus, the estimates by the Lyapunov dimension are of interest
(see, e.g. \cite{GrassbergerP-1983}).
Along with numerical methods for computing the Lyapunov dimension there
is an effective analytical approach, proposed by Leonov in 1991 \cite{Leonov-1991-Vest}
(see also \cite{LeonovB-1992,BoichenkoLR-2005,Kuznetsov-2016-PLA}).
The Leonov method is based
on a combination of the Douady-Oesterl\'{e} approach
with the direct Lyapunov method.
The advantage of the method is that it often allows one to estimate the Lyapunov dimension
of attractor without localization of attractor in the phase space and,
in many cases, to get an exact Lyapunov dimension formula
\cite{Leonov-2002,LeonovP-2005,Leonov-2012-PMM,LeonovPS-2013-PLA,
LeonovKKK-2015-arXiv-YangTigan,LeonovAK-2015,LeonovKKK-2016-CNSCS}.

Nowadays it is known various approaches to the definition of the Lyapunov dimension.
Below we use a rigorous definition \cite{Kuznetsov-2016-PLA}
of the Lyapunov dimension inspirited by
the Douady-Oesterl\'{e} theorem on the Hausdorff dimension of maps.
Consider an autonomous differential equation
\begin{equation} \label{eq:ode}
  \dot u = f(u),
\end{equation}
where $f: U \subseteq \mathbb{R}^n \to \mathbb{R}^n$ is a continuously differentiable vector-function,
$U$ is an open set.
Define by $u(t,u_0)$ a solution of \eqref{eq:ode} such that $u(0,u_0)=u_0$,
and consider the evolutionary operator $\varphi^t(u_0) = u(t,u_0)$.
We assume the uniqueness and existence of solutions of \eqref{eq:ode} for $t \in [0,+\infty)$.
Then system \eqref{eq:ode} generates a dynamical system
$\{\varphi^t\}_{t\geq0}$.
Let a nonempty set $K \subset U$ be invariant with respect to $\{\varphi^t\}_{t\geq0}$,
i.e. $\varphi^t(K) = K$ for all $t \geq 0$.
Consider the linearization of system (\ref{eq:ode})
along the solution $\varphi^t(u)$:
\begin{equation} \label{sfl}
  \begin{aligned}
    & \dot v = J(\varphi^t(u))v,  \quad J(u) = Df(u),
  \end{aligned}
\end{equation}
where $J(u)$ is the $n\times n$ Jacobian matrix,
the elements of which are continuous functions of $u$.
Suppose that $\det J(u) \neq 0 \quad \forall u \in U$.
Consider a fundamental matrix of linearized system \eqref{sfl}
$D\varphi^t(u)$ such that $D\varphi^0(u) = I$, where $I$ is a unit $n \times n$ matrix.
Let $\sigma_i(t,u) = \sigma_i(D\varphi^t(u))$, $i = 1,2,..\, n$, be the singular values of $D\varphi^t(u)$
with respect to their algebraic multiplicity
ordered so that $\sigma_1(t,u) \geq ... \geq \sigma_n(t,u) > 0$ for any $u \in U$ and $t \geq 0$.
The \emph{singular value function of order} $d \in [0,n]$ at $u \in U$ is defined as
\begin{equation}\label{defomega}
  \begin{cases}
  \omega_0(D\varphi^t(u)) = 1 , \\
  \omega_d(D\varphi^t(u)) =
  \sigma_1(t, u)\cdots\sigma_{\lfloor d \rfloor}(t, u)
    \sigma_{\lfloor d \rfloor+1}(t, u)^{d-\lfloor d \rfloor},\ \quad d \in (0,n), \\
  \omega_n(D\varphi^t(u)) =\sigma_1(t, u)\cdots\sigma_n(t, u),  \\
  \end{cases}
\end{equation}
where ${\lfloor d \rfloor}$ is the largest integer less or equal to $d$. 
For a certain moment of time $t$ \emph{the local Lyapunov dimension of the map $\varphi^t$
at the point $u \in \mathbb{R}^n$}
(or \emph{the finite-time local Lyapunov dimension of dynamical system $\{\varphi^t\}_{t\geq0}$})
is defined as \cite{Kuznetsov-2016-PLA}
\begin{equation}\label{locDOmaptmax}
  \dim_{\rm L}(\varphi^t,u) = \max\{d \in [0,n]: \omega_{d}(D\varphi^t(u)) \geq 1 \}
\end{equation}
and \emph{the Lyapunov dimension of the map  $\varphi^t$}
(or \emph{the finite-time Lyapunov dimension of dynamical system $\{\varphi^t\}_{t\geq0}$})
with respect to invariant set $K$ is defined as
\begin{equation}\label{DOmaptmax}
  \dim_{\rm L}(\varphi^t, K) = \sup\limits_{u \in K} \dim_{\rm L}(\varphi^t,u) =
  \sup\limits_{u \in K} \max\{d \in [0,n]: \omega_{d}(D\varphi^t(u)) \geq 1\}.
\end{equation}

The following is a corollary of the fundamental Douady--Oesterl\'{e} theorem \cite{DouadyO-1980}
\begin{theorem}\label{DOthm}
For any fixed $t > 0$
the Lyapunov dimension of the map $\varphi^t$ with respect
to a compact invariant set $K$, defined by \eqref{DOmaptmax}, is an upper
estimate of the Hausdorff dimension of the set $K$: $\dim_{\rm H}K \leq \dim_{\rm L}(\varphi^t, K)$.
\end{theorem}

For the estimation of the Hausdorff dimension of invariant compact set $K$
one can use the map $\varphi^t$ with any time $t$
(e.g. $t=0$ leads to the trivial estimate $\dim_{\rm H}K \leq n$),
therefore the best estimation is
\(
  \dim_{\rm H}{K} \le \inf_{t\geq0}\dim_{\rm L} (\varphi^t, K).
\)
By the properties of the singular value function and the cocycle property of fundamental matrix
we can prove \cite{Kuznetsov-2016-PLA}
that
\begin{equation}\label{DOlim}
  \inf_{t\geq0}\sup\limits_{u \in K} \dim_{\rm L}(\varphi^t,u)
  = \liminf_{t \to +\infty}\sup\limits_{u \in K} \dim_{\rm L}(\varphi^t,u).
\end{equation}
This property allows one to introduce
\emph{the Lyapunov dimension of dynamical system $\{\varphi^t\}_{t\geq0}$}
with respect to compact invariant set $K$
(often called \emph{the Lyapunov dimension of $K$})
as \cite{Kuznetsov-2016-PLA}
\begin{equation}\label{DOinf}
  \dim_{\rm L} (\{\varphi^t\}_{t \geq 0}, K) = \liminf_{t \to +\infty}\dim_{\rm L} (\varphi^t, K)
  = \liminf_{t \to +\infty}\sup\limits_{u \in K} \dim_{\rm L}(\varphi^t,u)
\end{equation}
which is an upper estimation of the Hausdorff dimension
\begin{equation}\label{HDOinf}
  \dim_{\rm H}{K} \le \dim_{\rm L} (\{\varphi^t\}_{t \geq 0}, K).
\end{equation}

Consider a set of finite-time Lyapunov exponents (of singular values) at the point $u$:
\begin{equation}\label{ftLE}
  \LEs_i(t,u) = \frac{1}{t}\ln\sigma_i(t,u), \ i=1,2,...,n \quad t > 0.
\end{equation}
Here the set $\{\LEs_i(t,u)\}_{i=1}^n$ is ordered by decreasing
(i.e. $\LEs_1(t,u) \geq \cdots \geq \LEs_n(t,u)$ for all $t>0$)
since the singular values are ordered by decreasing.
Define $j(t,u) = \max\{m: \sum_{i=1}^{m}\LEs_i(t,u) \geq 0\}$,
and let $n> j(t,u) \geq 1$.
Then the \emph{Kaplan-Yorke formula} \cite{KaplanY-1979} with respect to the finite-time
Lyapunov exponents $\{\LEs_i(t,u)\}_{i=1}^n$ is as follows \cite{Kuznetsov-2016-PLA}
\begin{equation}\label{lftKY}
   d_{\rm L}^{\rm KY}(\{\LEs_i(t,u)\}_{i=1}^n) =
   j(t,u) + \frac{\LEs_1(t,u) + \cdots + \LEs_{j(t,u)}(t,u)}{|\LEs_{j(t,u)+1}(t,u)|},
\end{equation}
and it coincides with the local Lyapunov dimension of the map $\varphi^t$ at the point $u$:
\[
   \dim_{\rm L}(\varphi^t, u) = d_{\rm L}^{\rm KY}(\{\LEs_i(t,u)\}_{i=1}^n).
\]
Thus, the use of Kaplan-Yorke formula \eqref{lftKY} with $\{\LEs_i(t,u)\}_{i=1}^n$
is rigorously justified by the Douady--Oesterl\'{e} theorem.
In the above formula if $n > \dim_{\rm L} (\varphi^t, u) > 1$,
 then for $j(t, u) = \lfloor \dim_{\rm L} (\varphi^t, u) \rfloor$ and $s(t, u) = \dim_{\rm L} (\varphi^t, u) - \lfloor \dim_{\rm L} (\varphi^t, u) \rfloor$
 from \eqref{locDOmaptmax} we have
\(
  0 = \frac{1}{t} \ln(\omega_{j(t,u)+s(t,u)}(D\varphi^t(u)))
    = \sum_{i=1}^{j(t,u)}\LEs_i(t,u) + s(t, u)\LEs_{j(t,u)+1}(t,u)
\).

It is known that
while the topological dimensions are invariant with respect to Lipschitz homeomorphisms,
the Hausdorff dimension is invariant with respect to Lipschitz diffeomorphisms
and the noninteger Hausdorff dimension is not invariant with respect to homeomor\-phisms \cite{HurewiczW-1941}.
Since the Lyapunov dimension is used as an upper estimate of the Hausdorff dimension,
its corresponding properties are important (see, e.g. \cite{OttWY-1984}).
Consider the dynamical system
$\big(\{\varphi^t\}_{t\geq0},(U\subseteq \mathbb{R}^n,||\cdot||) \big)$
under the smooth change of coordinates $w = h(u)$,
where $h: U \subseteq \mathbb{R}^n \to \mathbb{R}^n$ is a diffeomorphism.
In this case
the dynamical system
$\big(\{\varphi^t\}_{t\geq0},(U\subseteq \mathbb{R}^n,||\cdot||) \big)$
is transformed to
the dynamical system
$\big(\{\varphi_h^t\}_{t\geq0},(h(U)\subseteq \mathbb{R}^n,||\cdot||) \big)$,
and the compact set $K \subset U$ invariant with respect to $\{\varphi^t\}_{t\geq0}$
is mapped to the compact set $h(K) \subset h(U)$
invariant with respect to $\{\varphi_h^t\}_{t\geq0}$.

\begin{proposition}\label{thm:dDOunderdiff} (see, e.g. \cite{KuznetsovAL-2016,Kuznetsov-2016-PLA})
The Lyapunov dimension of the dynamical system $\{\varphi^t\}_{t\geq0}$
with respect to the compact invariant set $K$
is invariant with respect to any diffeomorphism
$h: U \subseteq \mathbb{R}^n \to \mathbb{R}^n$, i.e.
\begin{equation}\label{dDOunderdiff}
\begin{aligned}
  &
  \dim_{\rm L}(\{\varphi^t\}_{t\geq0},K)
  =
  \dim_{\rm L}(\{\varphi_h^t\}_{t\geq0},h(K)).
\end{aligned}
\end{equation}
\end{proposition}
\noindent This property and a proper choice of smooth change of coordinates
may significantly simplify
the computation of the Lyapunov dimension of dynamical system
(see also a discussion in \cite{LeonovK-2015-AMC}).

\subsection{Computation of the Lyapunov dimension}
For numerical computation of the finite-time Lyapunov exponents
there are developed various continuous and discrete algorithms
based on the singular value decomposition (SVD)
of fundamental matrix $D\varphi^t(u)$, which has the form
\(
  D\varphi^t(u)={\rm U}(t,u){\rm \Sigma}(t,u_0){\rm V}^*(t,u).
\)
Here ${\rm U}(t,u)^*{\rm U}(t,u) \equiv I \equiv {\rm V}(t,u)^*{\rm V}(t,u)$
and ${\rm \Sigma}(t)=\text{\rm diag}\{\sigma_1(t,u),...,\sigma_n(t,u)\}$
is a  diagonal matrix with positive real diagonal entries,
which are \emph{singular values} of $D\varphi^t(u)$
(thus the finite-time Lyapunov exponents
can be computed from ${\rm \Sigma}(t)$ according to \eqref{ftLE}).
An implementation of the discrete SVD method
for computing finite-time Lyapunov exponents $\{\LEs_i(t,u_0)\}_1^n$
in MATLAB can be found, e.g. in \cite{LeonovKM-2015-EPJST}.
It should be noted that some other algorithms
(e.g. Benettin's \cite{BenettinGGS-1980-Part2} and Wolf's \cite{WolfSSV-1985} algorithms),
widely used for the Lyapunov exponents computation,
are based on the so-called QR decomposition
and, in general, lead to the computation of the values
called \emph{finite-time Lyapunov exponents of the fundamental matrix columns}
$(y^1(t,u),...,y^n(t,u))=D\varphi^t(u)$
(or \emph{finite-time Lyapunov characteristic exponents}, LCEs)
at the point $u$ in which case
the set $\{\LCEs_i(t,u)\}_{i=1}^n$ ordered by decreasing for $t>0$
is defined as the set
$\{\frac{1}{t}\ln||y^i(t,u)||\}_{i=1}^n$.
The set $\{\LCEs_i(t,u)\}_{i=1}^n$ may significantly differ from the
$\{\LEs_i(t,u)\}_{i=1}^n$ and,
in the general,
$\dim_{\rm L}(\varphi^t, u) =d_{\rm L}^{\rm KY}(\{\LEs_i(t,u)\}_1^n) \neq d_{\rm L}^{\rm KY}(\{\LCEs_i(t,u)\}_1^n)$
\footnote{
In contrast to the definition of the Lyapunov exponents of singular values,
finite-time Lyapunov exponents of fundamental matrix columns
may be different for different fundamental matrices (see, e.g.  \cite{KuznetsovAL-2016}).
To get the set of all possible values of the Lyapunov exponents
of fundamental matrix columns
(the set with the minimal sum of values),
one has to consider the so-called normal fundamental matrices  \cite{Lyapunov-1892}.
Using, e.g, Courant-Fischer theorem, 
it is possible to show that
$\LCEs_1(t,u) =\LEs_1(t,u)$
and $\LEs_i(t,u) \leq \LCEs_i(t,u)$ for $1<i\leq n$.
For example, for the matrix  \cite{KuznetsovAL-2016}
 \(
    X(t)=\left(
      \begin{array}{cc}
        1 & g(t)-g^{-1}(t) \\
        0 & 1 \\
      \end{array}
    \right)
 \)
 we have the following ordered values:
 $  \LCEs_1(X(t)) =
  {\rm max}\big(\limsup\limits_{t \to +\infty}\frac{1}{t}\ln|g(t)|,
  \limsup\limits_{t \to +\infty}\frac{1}{t}\ln|g^{-1}(t)|\big),
  \LCEs_2(X(t)) = 0$;
 $
  \LEs_{1,2}(X(t)) = {\rm max, min}
  \big(
     \limsup\limits_{t \to +\infty}\frac{1}{t}\ln|g(t)|,
     \limsup\limits_{t \to +\infty}\frac{1}{t}\ln|g^{-1}(t)|
  \big).
 $
Thus, in general we have (see, e.g. discussion in  \cite{Kuznetsov-2016-PLA}):
\(
 \dim_{\rm L}(\varphi^t, u) =
 d_{\rm L}^{\rm KY}(\{\LEs_i(t,u)\}_1^n) \leq d_{\rm L}^{\rm KY}(\{\LCEs_i(t,u)\}_1^n).
\)
}.
Also there are known various examples in which the results
of Lyapunov exponents numerical computations
substantially differ from analytical results  \cite{TempkinY-2007,AugustovaBC-2015}.

Applying the statistical physics approach and assuming the ergodicity
(see, e.g.  \cite{KaplanY-1979,Ledrappier-1981,FredericksonKYY-1983,FarmerOY-1983}),
the Lyapunov dimension $\dim_{\rm L} (\{\varphi^t\}_{t \geq 0}, K)$
of attractor $K$ 
are often approximated by the local Lyapunov dimension $\dim_{\rm L} (\varphi^t, u_0)$
and its limit value $\limsup_{t\to+\infty}\dim_{\rm L} (\varphi^t, u_0)$
corresponding to a
trajectory $\{\varphi^t(u_0), t \geq 0 \}$
that 
belongs to the attractor ($u_0 \in K$).
However, from a practical point of view,
the rigorous proof of ergodicity is a challenging task
 \cite{BogoliubovK-1937,DellnitzJ-2002,Oseledec-1968,Ledrappier-1981}
(e.g. even for the well-studied Lorenz system),
which hardly can be done effectively in the general case
(see, e.g. the corresponding discussions
in \cite{BarreiraS-2000},\cite[p.118]{ChaosBook},\cite{OttY-2008},\cite[p.9]{Young-2013}
and the works  \cite{LeonovK-2007}
on the Perron effects of the largest Lyapunov exponent sign reversals).
An example of the effective rigorous use of the ergodic theory
for the estimation of the Hausdorff and Lyapunov dimensions is given, e.g. in  \cite{Schmeling-1998}.
Remark also that even if a numerical visualization of attractor $\widetilde{K}$
is obtained (which is only an approximation of the attractor $K$),
it is not clear how to choose a point on the attractor itself: $u_0 \in K$.

Thus, in general, to estimate the Lyapunov dimension of attractor $K$
according to \eqref{DOmaptmax}
we need \cite{KuznetsovMV-2014-CNSNS,LeonovKM-2015-EPJST}
to localize the attractor $K \subset K^{\varepsilon}$, consider a grid of points $K^{\varepsilon}_{\rm grid}$ on it, and
find the maximum  of corresponding finite-time local Lyapunov dimensions:
$\max\limits_{u \in K^{\varepsilon}_{\rm grid}} \dim_{\rm L}(\varphi^t,u)$.

To avoid the localization of attractors and numerical procedures,
we consider an effective analytical approach, proposed by Leonov
in 1991  \cite{Leonov-1991-Vest} (see also surveys  \cite{LeonovB-1992,Kuznetsov-2016-PLA}).
The Leonov method is based on a combination of the Douady-Oesterl\'{e} approach
with the direct Lyapunov method
and in the work  \cite{Kuznetsov-2016-PLA} it is shown how the method can be derived
from the invariance of the Lyapunov dimension of compact invariant set
with respect to the special smooth change of variables $h: U \subseteq \mathbb{R}^n \to \mathbb{R}^n$
with $Dh(u)=e^{V(u)(j+s)^{-1}}S$, where $V: U \subseteq \mathbb{R}^n \to \mathbb{R}^1$
is a differentiable scalar function
and $S$ is a nonsingular $n \times n$ matrix (see Proposition~\ref{thm:dDOunderdiff}).
Let $\lambda_i(u_0,S)$, $i=1,2,...,n,$ be the eigenvalues of the symmetrized Jacobian matrix
\begin{equation} \label{SJS}
  \frac{1}{2} \left( S J(u(t,u_0)) S^{-1} + (S J(u(t,u_0)) S^{-1})^{*}\right),
\end{equation}
ordered so that $\lambda_1(u_0, S) \ge \cdots \ge \lambda_n(u_0, S)$ for any $u_0 \in U$.

\begin{theorem}\label{theorem:th1}
Let $d=(j+s) \in [1,n]$, where
integer $j=\lfloor d \rfloor \in \{1,\ldots,n\}$
and real $s = (d - \lfloor d \rfloor) \in [0,1)$.
If there exist a differentiable scalar function $V: U \subseteq \mathbb{R}^n \to \mathbb{R}^1$
and a nonsingular $n\times n$ matrix $S$
such that
\begin{equation}\label{ineq:weilSVct}
  \sup_{u \in U} \big( \lambda_1 (u,S) + \cdots + \lambda_j (u,S)
  + s\lambda_{j+1}(u,S) + \dot{V}(u) \big) < 0,
\end{equation}
where $\dot{V} (u) = ({\rm grad}(V))^{*}f(u)$,
then for a compact invariant set $K\subset U$ we have
\[
   \dim_{\rm H}K \leq
    \dim_{\rm L}(\{\varphi^t\}_{t\geq0},K)
   \leq j+s.
\]
\end{theorem}

\noindent This theorem allows one to estimate
the singular values in the Lyapunov dimension
by the eigenvalues of symmetrized Jacobian matrix.
The proper choice of function $\dot{V}(u)$
allows one to simplify the estimation of the partial sum of eigenvalues and
the nonunitary nonsingular matrix $S$ (i.e. $S^{-1} \neq S^{*}$)
is used to make it possible the analytical computation
of the eigenvalues.
In Theorem~\ref{theorem:th1} the constancy of the signs of $V(u)$ or $\dot V(u)$ is not required.
A generalization of the above result for the discrete-time dynamical systems
can be found in  \cite{Kuznetsov-2016-PLA}.
Additionally,
if a localization of invariant set $K$ is known: $K \subset K^{\varepsilon} \subset U$,
then one can check \eqref{ineq:weilSVct} on $K^{\varepsilon}$ only.
Also we can consider
the Kaplan-Yorke formula with respect to the ordered set of eigenvalues of the symmetrized Jacobian matrix: 
\(
  d_{\rm L}^{\rm KY}(\{\lambda_{j}(u,S)\}_{i=1}^n),
\)
and its supremum on the set $K$ gives an upper estimation of the finite-time Lyapunov
dimension.

\begin{proposition}\label{thm:dLKYeig}
For a compact invariant set $K$ and any nonsingular $n\times n$ matrix $S$
we have
\begin{equation}\label{dLKYeig}
 \dim_{\rm H}K \leq \dim_{\rm L}(\{\varphi^t\}_{t\geq0},K) \leq
 \sup_{u \in K}d_{\rm L}^{\rm KY}\big(\{\lambda_{j}(u,S)\}_{i=1}^n\big).
\end{equation}
\end{proposition}
\noindent This is a generalization of ideas, discussed  e.g. in  \cite{DouadyO-1980,Smith-1986},
on the Hausdorff dimension estimation by the eigenvalues of symmetrized Jacobian matrix.




Since the function $u \mapsto \dim_{\rm L}(\varphi^t,u)$
is upper semi-continuous (see, e.g.  \cite[p.554]{Gelfert-2003}),
for each $t\geq 0$ there exists a critical point $u_{\rm L}(t) \in K$, which may be not unique,
such that $\sup_{u \in K}\dim_{\rm L}(\varphi^t,u)=\dim_{\rm L}(\varphi^t,u_{\rm L}(t))$.
An essential question (see discussion in  \cite[p.2146]{Kuznetsov-2016-PLA}) is
{\it whether there exists a critical path
$\gamma^{cr}=\{\varphi^t_{\rm GD}(u^{cr}_0),\ t\geq0 \}$
such that for each $t\geq0$
one of the corresponding critical points belongs to the critical path:
$\varphi^t_{\rm GD}(u^{cr}_0) = u_{\rm L}(t)$,
and, if so, whether the critical path is
an equilibrium or a periodic solution}.
The last part of the question was formulated
in  \cite[p.98, Question 1]{Eden-1990}\footnote{
  Another approach for the introduction of the Lyapunov dimension of dynamical system
  was developed by Constantin, Eden, Foia\c{s}, and Temam
   \cite{ConstantinFT-1985,Eden-1990,EdenFT-1991}.
  In the definition of the Lyapunov dimension of
  the dynamical system $\{\varphi^t\}_{t\geq0}$
  (see \eqref{DOmaptmax})
  they consider $\big(\omega_{d}(D\varphi^t(u))\big)^{1/t}$
  instead of $\omega_{d}(D\varphi^t(u))$
  and apply the theory of positive operators
  to prove the existence of a critical point
  $u^{cr}_{\rm E}$ (which may be not unique),
  where the corresponding global Lyapunov dimension
  achieves maximum (see  \cite{Eden-1990}):
  \(
  \dim_{\rm L}^{\rm E}(\{\varphi^t\}_{t\geq0},K)
  =\inf\{d \in [0,n]:
  \lim\limits_{t \to +\infty} \max\limits_{u \in K}\ln\big(\omega_{d}(D\varphi^t(u))\big)^{1/t}<0\} =
  \inf\{d \in [0,n]:
  \limsup\limits_{t \to +\infty}\ln\big(\omega_{d}(D\varphi^t(u^{cr}_{\rm E}))\big)^{1/t}<0\}
  =
  \dim_{\rm L}^{\rm E}(\{\varphi^t\}_{t\geq0},u^{cr}_{\rm E}),
  \)
  and, thus, rigorously justify the usage of the local Lyapunov dimension
  $\dim_{\rm L}^{\rm E}(\{\varphi^t\}_{t\geq0},u)$.
  Although it may seem that this definition allows to reduce
  computational complexity
  (since the supremum over the set $K$ has to be computed only once for $t=+\infty$)
  as compared with the definition of \eqref{DOinf}
  (where the supremum has to be computed for each $t \in (0,+\infty)$),
  it does not have a clear sense for a finite-time interval $(0,T)$,
  which can only be considered in numerical experiments.
  Remark also that $\dim_{\rm L}(\varphi^t, K)$,
  according to the Douady--Oesterl\'{e} theorem,
  has clear sense for any fixed $t$ and, thus,
  in numerical experiments it can be computed, according to \eqref{DOlim},
  only for sufficiently large time $t=T$
  (i.e the supremum over the set $K$ is computed only once for $t=T$).
}.
A \emph{conjecture on the Lyapunov dimension of self-excited attractors} \cite{KuznetsovL-2016-ArXiv} is that \emph{for ''typical'' systems
the Lyapunov dimension of self-excited attractor
is less then the Lyapunov dimension of one of the unstable equilibria,
the unstable manifold of which intersects with the basin of attraction}.
Next corollary addresses the question and conjecture
and is used to get the \emph{exact Lyapunov dimension}
(this term was suggested in  \cite{DoeringGHN-1987})
for the global attractors, which involves equilibria.

\begin{corollary} \label{cor:equil}
  If for $d=j+s$,  defined by Theorem~\ref{theorem:th1}
  (i.e. for $d: \dim_{\rm H}K \leq d$),
  at an equilibrium point $u_{eq}^{cr}$
  ($u_{eq}^{cr} = \varphi^{t}(u_{eq}^{cr})$ for any $t \geq 0$)
  the relation
  \[
    \dim_{\rm L}(\{\varphi^t\}_{t\geq0},u_{eq}^{cr}) = j+s
  \]
  holds, then
  for any compact invariant set $K \supset u_{eq}^{cr}$
  we get the exact Lyapunov dimension formula
  \[
  \begin{aligned}
   \dim_{\rm L}(\{\varphi^t\}_{t\geq0},K)
   =\dim_{\rm L}(\{\varphi^t\}_{t\geq0},u_{eq}^{cr}) = j + s.
  \end{aligned}
  \]
\end{corollary}

Next statement is used to compute the \emph{Lyapunov dimension}
at an equilibrium with the help of the corresponding eigenvalues.
\begin{proposition}\label{ueqcr}
Suppose that at one of the equilibrium points of the dynamical system $\{\varphi^t\}_{t\geq0}$:
$u_{eq}\equiv\varphi^t(u_{eq})$, $u_{eq} \in U$,
the matrix $J(u_{eq})$ has simple real eigenvalues:
$\{\lambda_i(u_{eq})\}_{i=1}^{n}$, $\lambda_{i}(u_{eq}) \geq \lambda_{i+1}(u_{eq})$.
Then
\[
  \dim_{\rm L}(\{\varphi^t\}_{t\geq0},u_{eq})
  =
  d_{\rm L}^{\rm KY}(\{\lambda_i(u_{eq})\}_{i=1}^n).
\]
\end{proposition}
\noindent The proof follows from the invariance of the Lyapunov dimension
and the fact that in this case there exists a nonsingular matrix $S$ such that
\(
  SJ(u_{eq})S^{-1}={\rm diag}\big(\lambda_1(u_{eq}),..,\lambda_n(u_{eq})\big)
\)
and $\dim_{\rm L}(S\varphi^t,Su_{eq}) \equiv d_{\rm L}^{\rm KY}(\{\lambda_i(u_{eq})\}_{i=1}^n)$
for any $t>0$.

For the study of continuous-time dynamical system in $\mathbb{R}^3$,
which possesses an absorbing ball
(i.e. dissipative in the sense of Levinson),
the following result \cite{Leonov-1991-Vest} is useful.
Consider a certain open set $K^{\varepsilon} \subset U$, which is diffeomorphic
to a ball, whose boundary $\partial \overline{K^{\varepsilon}}$
is transversal to the vectors $f(u)$, $u \in \partial \overline{K^{\varepsilon}}$.
Let the set $K^{\varepsilon}$ be a positively invariant for the solutions of system \eqref{eq:ode},
i.e. $\varphi^t(K^{\varepsilon}) \subset K^{\varepsilon}$, $t \geq 0$.
\begin{theorem}
\label{theorem:th2}
Suppose, a continuously differentiable function $V: U \subseteq \mathbb{R}^3 \to \mathbb{R}^1$ and
a non-degenerate matrix $S$ exist such that
\begin{equation}\label{pmm36}
  \sup_{u \in K^{\varepsilon}} \big(
  \lambda_1(u,S)+\lambda_2(u,S)+\dot V(u)
  \big)
  <0.
\end{equation}
Then $\varphi^t(u_0)$ with any initial data $u_0\in K^{\varepsilon}$
tends to the stationary set of dynamical system $\{\varphi^t\}_{t\geq0}$ as $t\to+\infty$.
\end{theorem}
\noindent In this case the minimal attracting invariant set $K \subset K^{\varepsilon}$
consists of equilibria and in the case of a finite set of equilibrium points in the system
we have $\dim_{\rm H}K=0$.

\section{Main results. Analytical estimations of the Lyapunov dimension of G-D system}

Let $u = (x,\, y, \, z) \in U = \mathbb{R}^3$,
$\{\varphi_{\rm GD}^t\}_{t \geq 0}$ is the dynamical system generated by
\eqref{sys:lorenz-general} with positive parameters $\sigma$,$r$,$A$,$b$,
and $K \subset \mathbb{R}^3$ is a compact invariant set of $\{\varphi_{\rm GD}^t\}_{t \geq 0}$.

By Theorems \ref{theorem:th1} and \ref{theorem:th2}
it can be formulated the assertion on
the Lyapunov dimension of $\{\varphi_{\rm GD}^t\}_{t \geq 0}$.

\begin{theorem}\label{theorem:th3}
Suppose that either the inequality $b < 1$ or the inequalities
$b \geqslant 1$, $\sigma > b$ are valid. \\
If
\begin{equation}
  \left(r+\frac{\sigma}{A}\right)^2 < \frac{(b + 1)(b +
\sigma)}{A}, \label{cond:param1}
\end{equation}
then $\varphi_{\rm GD}^t(u)$ with any $u \in U$
tends to an equilibrium as $t \rightarrow +\infty$
(i.e. the minimal attractive set of $\{\varphi_{\rm GD}^t\}_{t \geq 0}$
is a set of equilibria and its Hausdorff dimension is zero). \\
If
\begin{equation}
  \left(r+\frac{\sigma}{A}\right)^2 \geq \frac{(b + 1)(b +
\sigma)}{A}, \label{cond:param2}
\end{equation}
then for any compact invariant set $K$ of $\{\varphi^t_{\rm GD}\}_{t \geq 0}$ we have
\begin{equation}
  \dim_{\rm L} (\{\varphi^t_{\rm GD}\}_{t \geq 0}, K) \leq 3 - \frac{2(\sigma + b + 1)}{\sigma + 1 +
\sqrt{(\sigma - 1)^2 + A\left(\frac{\sigma}{A}+r\right)^2}}.
  \label{ineq:lyap-dim}
\end{equation}

\end{theorem}

\begin{proof}
The Jacobian matrix for system~\eqref{sys:lorenz-general} is as follows
\begin{equation}\label{eq:J}
	J(u) = \left(
		\begin{array}{ccc}
			-\sigma & \sigma - Az & -Ay \\
			r - z & -1 & -x \\
			y & x & -b
		\end{array}
	\right).
\end{equation}
Consider a matrix
\begin{equation*}
	S = \left(
		\begin{array}{ccc}
			-\frac{1}{\sqrt{A}} & 0 & 0 \\
			0 & 1 & 0 \\
			0 & 0 & 1
		\end{array}
	\right).
\end{equation*}

Then
\begin{equation}\label{matrix:sjs}
  \frac{1}{2} \left( S J(u) S^{-1} + (S J(u) S^{-1})^{\ast}\right) =
\left(
    \begin{array}{ccc}
      -\sigma & -\sqrt{A}z - \frac{\sigma + Ar}{2\sqrt{A}} & 0 \\
      -\sqrt{A}z - \frac{\sigma + Ar}{2\sqrt{A}} & -1 & 0 \\
      0 & 0 & -b
    \end{array}
  \right),
\end{equation}
and its characteristic polynomial
has the form
\begin{equation*}
	(\lambda + b)\left[ \lambda^2 + (\sigma + 1)\lambda + \sigma -
\left( \sqrt{A}z + \frac{\sigma + Ar}{2\sqrt{A}} \right)^2 \right].
\end{equation*}
Denote by $\lambda_i = \lambda_i(u, S)$, $i = 1,2,3$,
the eigenvalues of matrix~\eqref{matrix:sjs}.
Then
\begin{eqnarray*}
\lambda_2 & = & -b, \\
\lambda_{1,3} & = & -\frac{(\sigma + 1)}{2} \pm
\frac{1}{2}\sqrt{(\sigma - 1)^2
+ 4 \left( \sqrt{A}z + \frac{\sigma + Ar}{2\sqrt{A}} \right)^2}.
\end{eqnarray*}
Thus, $\lambda_1 > \lambda_3$ and $\lambda_3 < 0$.
Let us find the conditions under which the
inequality $\lambda_2 > \lambda_3$ holds, i.e.
\begin{equation}\label{eq:lam2lam3}
\lambda_2 - \lambda_3
= -b + \frac{1}{2}(\sigma + 1) + \frac{1}{2}\sqrt{(\sigma - 1)^2
+ 4\left( \sqrt{A}z + \frac{\sigma + Ar}{2\sqrt{A}} \right)^2} > 0.
\end{equation}
If $(\sigma + 1) > 2b$, then inequality \eqref{eq:lam2lam3} is valid.
If $(\sigma + 1) \leqslant 2b$, then inequality \eqref{eq:lam2lam3} is
equivalent to the following relation
\begin{eqnarray*}
(\sigma + 1)^2 + 4b^2 - 4b(\sigma + 1) &<& (\sigma - 1)^2  +
4\left(\sqrt{A}z + \frac{\sigma + Ar}{2\sqrt{A}} \right)^2 \\
\Leftrightarrow \sigma - b\sigma - b + b^2 &<&
\left(\sqrt{A}z + \frac{\sigma + Ar}{2\sqrt{A}} \right)^2
\Leftrightarrow (\sigma -b)(1-b) <
\left(\sqrt{A}z + \frac{\sigma + Ar}{2\sqrt{A}} \right)^2.
\end{eqnarray*}
The latter is true in the case when $(\sigma - b)(1 - b) < 0$.
Hence, if the inequalities
\begin{equation}
(\sigma + 1) > 2 b \quad \text{or} \quad
\left\{
\begin{array}{l}
(\sigma + 1) \leqslant 2b, \\
(\sigma - b)(1 - b) < 0
\end{array}
\right.
\label{cond:lambda_sorted1}
\end{equation}
hold, then $\lambda_2 > \lambda_3$.
Inequalities~\eqref{cond:lambda_sorted1} are equivalent 
to the following expressions
\begin{equation}
b < 1  \quad \text{or} \quad
\left\{
\begin{array}{l}
b \geqslant 1, \\
\sigma > b
\end{array}
\right.
\label{cond:lambda_sorted2}
\end{equation}
and the conditions of Theorem~\ref{theorem:th3} are fulfilled.
This implies that under these conditions $\lambda_3$ is the
smallest eigenvalue.

Consider $s \in [0,\,1)$ and the following relations
\begin{eqnarray*}
2 (\lambda_1 + \lambda_2 + s \lambda_3) & = & -(\sigma + 1 + 2b) -
s(\sigma + 1) +
(1-s)\sqrt{(\sigma - 1)^2  + 4\left(\sqrt{A}z + \frac{\sigma + Ar}{2\sqrt{A}}\right)^2} \\
& \leqslant & -(\sigma + 1 + 2b) - s(\sigma + 1) + (1-s)\left[(\sigma
- 1)^2 + A\left(\frac{\sigma}{A}+r\right)^2\right]^\frac{1}{2} +\\
&& + \frac{2(1-s)}{\left[(\sigma - 1)^2 +
A\left(\frac{\sigma}{A}+r\right)^2\right]^\frac{1}{2}} \left[ -
(\sigma + Ar)z + Az^2\right].
\end{eqnarray*}

Denote
\begin{equation}
	w(x,y,z) = -(\sigma + Ar)z + Az^2.
\end{equation}

Choose the Lyapunov-like function $V(x,y,z)$ as follows
\begin{equation*}
	V(x,y,z) = \frac{2(1-s)}{\left[(\sigma - 1)^2
			+ A\left(\frac{\sigma}{A}+r\right)^2 \right]^{\frac{1}{2}}} \, \vartheta(x,y,z),
\end{equation*}
where	$\vartheta(x,y,z) = \gamma_1 x^2 + \gamma_2 y^2 + \gamma_3 z^2 + \gamma_4 z$
and $\gamma_1$, $\gamma_2$, $\gamma_3$, $\gamma_4$ are varying parameters.

Differentiation of $\vartheta$ along solutions of system~\eqref{sys:lorenz-general}
yields
\begin{align}
& \dot{\vartheta}(x,y,z) = 2 \gamma_1 x (\sigma (y-x) - Ayz) +
2 \gamma_2 y(r x - y - x z)  + (2 \gamma_3 z + \gamma_4) (-b z + x y) & \nonumber\\
& = 2(\gamma_3 - \gamma_2 - \gamma_1 A ) xyz - 2\gamma_1 \sigma x^2
+ (2\gamma_1 \sigma + \gamma_4 + 2\gamma_2 r) xy
- 2\gamma_2 y^2 - 2\gamma_3 b z^2 - \gamma_4 b z. & \label{eq:deriv-sys}
\end{align}
Thus
\begin{eqnarray*}
\dot{\vartheta}(x,y,z) + w(x,y,z) & = &
2(\gamma_3 - \gamma_2 - \gamma_1 A ) xyz +
(2\gamma_1 \sigma + \gamma_4 + 2\gamma_2 r) xy - 2\gamma_2 y^2 -\\
&& - 2\gamma_1 \sigma x^2 + (A - 2\gamma_3 b ) z^2 - (\sigma + Ar + \gamma_4 b) z.
\end{eqnarray*}
Parameters $\gamma_1$, $\gamma_2$, $\gamma_3$, $\gamma_4$ are
chosen such that $\dot{\vartheta}(x,y,z) + w(x,y,z)$
takes the form of polynomial
\begin{equation}
	A_1 x^2 + A_2 x y +A_3 y^2 + B_1 z^2,
\label{polinom}
\end{equation}
i.e. the coefficients of monomials $x y z$ and $z$
in~\eqref{eq:deriv-sys} are zero and
\begin{equation}\label{sys:equiv1}
\gamma_3 = A \gamma_1 + \gamma_2, \quad
\gamma_4 = - \frac{\sigma + Ar}{b}.
\end{equation}
From~\eqref{sys:equiv1} we have
\begin{eqnarray*}
\dot{\vartheta}(x,y,z)  + w(x,y,z) & = & - 2\gamma_1 \sigma x^2 + (2\gamma_1 \sigma -
\frac{\sigma + Ar}{b} + 2\gamma_2 r) xy - 2\gamma_2 y^2 +
(-2 b (A \gamma_1 + \gamma_2) + A) z^2 \\
& = & A_1 x^2 + A_2 xy + A_3 y^2 + B_1 z^2,
\end{eqnarray*}
where
\begin{equation}
A_1 = -2 \gamma_1 \sigma, \,\,\,
A_2 = 2 \gamma_1 \sigma -  \frac{\sigma + Ar}{b} + 2\gamma_2 r, \,\,\,
A_3 = - 2\gamma_2, \,\,\,
B_1 = -2 b (A \gamma_1  + \gamma_2) + A.
\label{def:coef}
\end{equation}
Polynomial~\eqref{polinom} can be written as
$$
	A_3 \left(y + \frac{A_2}{2 A_3} x \right)^2 +
	\frac{4 A_1 A_3 - A_2^2}{4 A_3} x^2  + B_1 z^2 .
$$
Hence the inequality $\dot{\vartheta}(x,y,z)  + w(x,y,z) \leqslant 0$ holds if and only if
\begin{equation}
		B_1 \leqslant 0, \quad
		A_3 < 0, \quad
		A_1 A_3 - \left(\frac{A_2^2}{2}\right)^2 \geqslant 0 . \label{sys:ineq1}
\end{equation}

Combining~\eqref{def:coef} with~\eqref{sys:ineq1}, we obtain
\begin{eqnarray}
&& \gamma_2 > 0, \label{ineq1}\\
&& \gamma_1 \geqslant \frac{1}{2b} - \frac{1}{A}\gamma_2,
\label{ineq2}\\
&& 4\sigma \gamma_1 \gamma_2 - \left(\gamma_1 \sigma -
\frac{\sigma + Ar}{2b} + \gamma_2 r \right)^2 \geqslant 0.
\label{ineq3}
\end{eqnarray}
Inequality \eqref{ineq3} is solvable with respect to $\gamma_1$ if and only if
its discriminant is nonnegative:
\[
  D = 8 \sigma^2 \gamma_2 \left( 2\gamma_2 (1-r) + \frac{\sigma +
  Ar}{b}\right) \geqslant 0.
\]
By \eqref{ineq1} and since $r > 1$,
the latter is equivalent to the following relation
\begin{equation}
2\gamma_2 (r-1) \leqslant \frac{\sigma +
Ar}{b}
\quad \Leftrightarrow \quad \gamma_2
\leqslant \frac{\sigma + Ar}{2 (r-1) b}.
\label{cond:solvable}
\end{equation}

Hence if condition \eqref{cond:solvable} holds, then \eqref{ineq3} is
equivalent to the
relation  $\Gamma_{-}(\gamma_2) \leqslant \gamma_1 \leqslant \Gamma_{+}(\gamma_2)$, where
\begin{equation}
	\Gamma_{\pm}(\gamma_2) = \frac{1}{\sigma}\left[\gamma_2(2-r)+\frac{\sigma
+ Ar}{2b} \pm \sqrt{2\gamma_2 \left(2(1-r)\gamma_2 +
\frac{\sigma + Ar}{b}\right)}\right]
\end{equation}
are the roots of quadratic polynomial in the left-hand side of
\eqref{ineq3}.

Consider now the location of the roots $\Gamma_{\pm}(\gamma_2)$ on the
real axis.
If $1 < r \leqslant 2$, then by \eqref{ineq1} we have $\gamma_2(2-r)+\frac{\sigma + Ar}{2b}
\geqslant 0$.
If $r > 2$, then by \eqref{cond:solvable} the relation
\begin{equation*}
\gamma_2(2-r)+\frac{\sigma + Ar}{2b} \geqslant 0 \quad \Leftrightarrow
\quad \gamma_2 \leqslant \frac{\sigma + Ar}{2(r-2)b}
\end{equation*}
holds since
\begin{equation*}
\frac{\sigma + Ar}{2(r-1)b} < \frac{\sigma + Ar}{2(r-2)b}.
\end{equation*}
Thus, for $\gamma_2$, defined by \eqref{ineq1} and \eqref{cond:solvable},
and $r > 1$ we have $\gamma_2(2-r)+\frac{\sigma + Ar}{2b} \geqslant 0$ and $\Gamma_{+}(\gamma_2) \geqslant 0$.

Let us prove that for $\gamma_2$, defined by \eqref{ineq1} and \eqref{cond:solvable},
we have $\Gamma_{-}(\gamma_2) \geqslant 0$.
It is true since
\begin{eqnarray*}
\gamma_2(2-r)+\frac{\sigma + Ar}{2b} &\geqslant& \sqrt{2\gamma_2
\left(2(1-r)\gamma_2 +
\frac{\sigma + Ar}{b}\right)} \\
\Leftrightarrow \left(\gamma_2(2-r)+\frac{\sigma +
Ar}{2b}\right)^2 &\geqslant& 2\gamma_2 \left(2(1-r)\gamma_2 +
\frac{\sigma + Ar}{b}\right) \\
\Leftrightarrow r^2\gamma_2^2 - \frac{\sigma + Ar}{b}r\gamma_2 +
\left(\frac{\sigma + Ar}{2b}\right)^2 &=&
\left(r\gamma_2 - \frac{\sigma + Ar}{2b}\right)^2 \geqslant 0.
\end{eqnarray*}
Hence if $r > 1$, then $0 \leqslant \Gamma_{-}(\gamma_2) \leqslant \gamma_1
\leqslant \Gamma_{+}(\gamma_2)$.

Let $\Gamma (\gamma_2) = \frac{1}{2b} - \frac{1}{A}\gamma_2$.
Thus if $r > 1$, the conditions \eqref{cond:lambda_sorted2} holds,
and there exist nonnegative $\gamma_1,\gamma_2$
such that a system of inequalities
\begin{equation}
\begin{cases}
0 < \gamma_2 \leqslant \frac{\sigma + Ar}{2 (r-1)b}, \\
\max \left\{ \Gamma(\gamma_2), \Gamma_{-}(\gamma_2) \right\} \leqslant
\gamma_1 \leqslant \Gamma_{+}(\gamma_2)
\end{cases} \label{cond:gamma}
\end{equation}
is solvable, then the inequality $\dot{\vartheta} + w \leqslant 0$ is valid.


Let us show that system \eqref{cond:gamma} is solvable.
Note that
\begin{eqnarray*}
&& \Gamma_{+} (0) = \Gamma_{-} (0) = \frac{\sigma + Ar}{2 b} > 0,\\
&& \Gamma_{+} \left(\frac{\sigma + Ar}{2 b(r-1)}\right) =
\Gamma_{-} \left(\frac{\sigma + Ar}{2 b(r-1)}\right) =
\frac{\sigma + Ar}{2b(r-1)\sigma} > 0 ,\\
&& \Gamma (0) = \frac{1}{2b} > 0 , \quad
\Gamma \left(\frac{\sigma + Ar}{2 b(r-1)}\right) = - \frac{\sigma
+A}{2b(r-1)A} < 0.
\end{eqnarray*}

This implies that the upper half plane, defined by the inequality
$\gamma_1 \geqslant \Gamma (\gamma_2)$, always
intersects the domain bounded by the curves $\Gamma_{\pm}(\gamma_2)$.
This intersection corresponds to
the existence domain of solutions of system \eqref{cond:gamma}.

Thus, for the chosen matrix $S$ and Lyapunov-like function
$V(x,y,z)$ if \eqref{cond:param2} is valid and
\begin{equation}
	s > \frac{-(\sigma + 1 + 2 b) + \sqrt{(\sigma - 1)^2 +
A\left(\frac{\sigma}{A}+r\right)^2}}
	{\sigma + 1 + \sqrt{(\sigma - 1)^2 +
A\left(\frac{\sigma}{A}+r\right)^2}}, \label{ineq:s}
\end{equation}
then for system \eqref{sys:lorenz-general} the conditions of Theorem \ref{theorem:th1} hold.
If \eqref{cond:param1} is valid and $s = 0$, then for system \eqref{sys:lorenz-general}
the conditions of Theorem \ref{theorem:th2} hold.

Hence if~\eqref{cond:param2} holds, then
$\dim_{\rm L} (\{\varphi^t_{\rm GD}\}_{t \geq 0}, K) \leqslant 2 + s$ for all $s$,
satisfying~\eqref{ineq:s}.
This implies inequality~\eqref{ineq:lyap-dim}
and completes the proof of the theorem.
\end{proof}

For system \eqref{sys:lorenz-general} with physically sounded value of parameter $b = 1$,
the upper estimate \eqref{ineq:lyap-dim} can be improved  \cite{LeonovM-2016-MAIK}.
\begin{theorem}\label{theorem:th4}
Let $b = 1$ and $\sigma \geq A r$.
If
\begin{equation}\label{th4:cond-main}
  \frac{2(\sigma + 2)}{\sigma + 1 + \sqrt{(\sigma - 1)^2 + 4 \sigma r}} \leq 1,
\end{equation}
then
\begin{equation}\label{estim}
  \dim_{\rm L} (\{\varphi^t_{\rm GD}\}_{t \geq 0}, K)
  \leq 3 - \frac{2(\sigma + 2)}{\sigma + 1 + \sqrt{(\sigma - 1)^2 + 4 \sigma r}}.
\end{equation}
\end{theorem}

\begin{proof}

Here we use the following idea suggested by Leonov \cite{Leonov-2016-ArXiv}.
The relation
\begin{equation}\label{cond:matrixQ}
  \frac{1}{2} \left( S J(u) S^{-1} + (S J(u) S^{-1})^{\ast}\right) + \mu I > 0,
\end{equation}
where $\mu = \frac{-\Tr J(u)}{1-s} = \frac{\sigma+b+1}{1-s} = \frac{\sigma+2}{1-s}$ and $s \in [0,1)$
is equivalent to condition \eqref{ineq:weilSVct} of Theorem \ref{theorem:th1} with $V \equiv 0$.
Note that $\Tr J(u) = \lambda_1(u,S) + \lambda_2(u,S) +\lambda_3(u,S)$.
The positive definiteness of matrix \eqref{cond:matrixQ} means
that $\lambda_3(u,S) + \mu > 0 \Leftrightarrow \lambda_1(u,S) + \lambda_2(u,S) +s \lambda_3(u,S) < 0$.

Consider a matrix
\begin{equation*}
 S = \left(
 \begin{array}{ccc}
 \sqrt{\frac{r}{\sigma}} & 0 & 0 \\
 0 & 1 & 0 \\
 0 & 0 & 1
 \end{array}
 \right).
\end{equation*}

Condition \eqref{cond:matrixQ} means that all leading principal minors $\Delta_{1,2,3}$
of the corresponding matrix are positive.
For the chosen matrix $S$ we have $\Delta_1 = -\sigma + \mu > 0$, $\Delta_3 = (-1 + \mu) \, \Delta_2$ and
relation \eqref{cond:matrixQ} can be expressed in the following way
\begin{equation}\label{cond:detQ}
  \Delta_3 = \left|
    \begin{array}{ccc}
    -\sigma + \mu & \sqrt{r \sigma} - \frac{z}{2}\sqrt{\frac{\sigma}{r}}\left(1 + \frac{Ar}{\sigma}\right) \, &
    \,\frac{y}{2}\sqrt{\frac{\sigma}{r}}\left(1 - \frac{Ar}{\sigma}\right) \\
    \sqrt{r \sigma} - \frac{z}{2}\sqrt{\frac{\sigma}{r}}\left(1 + \frac{Ar}{\sigma}\right) & -1 + \mu& 0 \\
    \frac{y}{2}\sqrt{\frac{\sigma}{r}}\left(1 - \frac{Ar}{\sigma}\right) & 0 & -1+ \mu
    \end{array}
  \right| \, > \, 0.
\end{equation}
Condition \eqref{cond:detQ} can be rewritten as
\[
  \frac{r}{\sigma}\left(\left(\mu -\sigma\right) \left(\mu - d\right) - r \sigma \right)
  + \left(1 + {\frac{Ar}{\sigma}}\right) \left( -\frac{{z}^{2}}{4} \left( 1 + {\frac{Ar}{\sigma}} \right)
  - \frac{y^2}{4} \, \frac{\left(1 - \frac{Ar}{\sigma} \right)^2}{\left(1 + \frac{Ar}{\sigma}\right)}
  + r z\right) > 0.
\]
One can see that for $b = 1$, $\sigma \geq Ar$ and $V \equiv 0$ condition \eqref{ineq:weilSVct}
holds for all $(x, y, z)$ from $\Omega(x, y, z)$ (see \eqref{localize_set}) if
\begin{equation}
  r \leq \frac{(\mu - \sigma)(\mu - 1)}{\sigma}.
\end{equation}
The expression $r = \frac{(\mu - \sigma)(\mu - 1)}{\sigma}$ is equivalent to the relation
\begin{equation*}
  \frac{2(\sigma + 2)}{\sigma + 1 + \sqrt{(\sigma - 1)^2 + 4 \sigma r}} = 1 - s.
\end{equation*}
Thus, if \eqref{th4:cond-main} is valid, then all the conditions of Theorem \ref{theorem:th1}
for system \eqref{sys:lorenz-general} hold.
\end{proof}

The obtained result is a development
of results from \cite{LeonovKM-2015-EPJST,LeonovM-2016-MAIK}
for all values of parameters for which
the transformation of system \eqref{sys:lorenz-general}
to \eqref{sys:conv_fluid} is valid (see conditions \eqref{sys:gen-lorenz:change_var}).

Theorems~\ref{theorem:th3} and \ref{theorem:th4} imply the following
\begin{conseq}\label{conseq}
If
\begin{enumerate}[label=(\roman*)]
  \item $\sigma = Ar$, $b < 1$ \label{cor:2:cond1} or
  \item $\sigma = Ar$, $b \geq 1$, $\sigma > b$ \label{cor:2:cond1_2} or
  \item $\sigma \geq Ar$, $b = 1$ \label{cor:2:cond2}
\end{enumerate}
and
\[
  \frac{2(\sigma + b + 1)}{\sigma + 1 + \sqrt{(\sigma - 1)^2 + 4 \sigma r}} \leq 1,
\]
then the Lyapunov dimension of the zero equilibrium of $\{\varphi^t_{\rm GD}\}_{t \geq 0}$
coincides with \eqref{estim}
and for any compact invariant set
$K \supset S_0 = (0,0,0)$ we get the exact Lyapunov dimension formula
\begin{equation}\label{finalLDform}
  \dim_{\rm L} (\{\varphi^t_{\rm GD}\}_{t \geq 0}, K) =
  \dim_{\rm L} (\{\varphi^t_{\rm GD}\}_{t \geq 0}, S_0) =
  3 - \frac{2(\sigma + b + 1)}{\sigma + 1 + \sqrt{(\sigma - 1)^2 + 4 \sigma r}}.
\end{equation}
\end{conseq}
\begin{proof}
The Jacobi matrix $J(u)$ from equation \eqref{eq:J}
at equilibrium $u = S_0$
has the following simple real eigenvalues
\begin{equation}\label{eq:JS0:eigVals}
  \lambda_{1,3}(S_0) = \frac{1}{2}\left(-(\sigma + 1) \pm \sqrt{(\sigma-1)^2 + 4\sigma r}\right),
  \quad \lambda_2(S_0) = -b.
\end{equation}
For $r > 1$, we have $\lambda_1(S_0) > 0$, $\lambda_{2,3}(S_0) < 0$.
If \ref{cor:2:cond1} or \ref{cor:2:cond2}, then it follows that
$\lambda_2(S_0) > \lambda_3(S_0)$,
$\lambda_1(S_0) + \lambda_2(S_0) + \lambda_3(S_0) = -\sigma - b -1 < 0$
and
from \eqref{th4:cond-main}
it follows that $\lambda_1(S_0) + \lambda_2(S_0) \geq 0$.
Then according to \eqref{lftKY} we have
\begin{align*}
  d_{\rm L}^{\rm KY}(\{\lambda_i(S_0)\}_{i=1}^n) &= 2 + \frac{\lambda_1(S_0) + \lambda_2(S_0)}{|\lambda_3(S_0)|} =
  2 + \frac{-(\sigma + 1 -2 b) + \sqrt{(\sigma - 1)^2 + 4 \sigma r}}{\sigma + 1 + \sqrt{(\sigma - 1)^2 + 4 \sigma r}} =
  &\\
  & = 3 - \frac{2(\sigma + b + 1)}{\sigma + 1 + \sqrt{(\sigma - 1)^2 + 4 \sigma r}}.&
\end{align*}
By Proposition \ref{ueqcr}
\[
    \dim_{\rm L}(\{\varphi^t_{\rm GD}\}_{t\geq0},S_0) =
    d_{\rm L}^{\rm KY}(\{\lambda_i(S_0)\}_{i=1}^n) =
    3 - \frac{2(\sigma + b + 1)}{\sigma + 1 + \sqrt{(\sigma - 1)^2 + 4 \sigma r}}
\]
and according to Corollary \ref{cor:equil} for any compact invariant set
$K \supset S_0$ we get \eqref{finalLDform}.
\end{proof}

Note that formula \eqref{finalLDform} coincides with the exact Lyapunov dimension formula
for the classical Lorenz system~ \cite{Leonov-2002,LeonovK-2015-AMC,LeonovKKK-2016-CNSCS}.
In the Lorenz system the maximum of the local Lyapunov dimensions
is also achieved at the zero equilibrium
and this fact is known as the so-called \emph{Eden conjecture on the Lorenz system}
 \cite{Eden-1989-PhD,Eden-1990,LeonovL-1993,DoeringG-1995}.
The main direction of its further study is to extend the domain of parameters
for which the exact Lyapunov dimension formula for the Lorenz system is valid.

\section{Numerical experiments and discussion of the results}
Below we consider the dynamical system $\{\varphi^t_{\rm GD}\}_{t\geq0}$,
generated by the generalized Lorenz system \eqref{sys:lorenz-general},
various types of its attractors $K$, and their Lyapunov dimensions.
Here $\varphi^t_{\rm GD}\big((x_0,y_0,z_0)\big)$
is a solution of \eqref{sys:lorenz-general} with the initial condition $(x_0,y_0,z_0)$, i.e.
\[
  \varphi^t_{\rm GD}\big((x_0,y_0,z_0)\big) = \big(x(t,(x_0,y_0,z_0)),y(t,(x_0,y_0,z_0)),z(t,(x_0,y_0,z_0)) \big).
\]

Let $\rho(K, u) = \inf_{v \in K} ||v - u||$
be the distance from the point $u \in U$ to the set $K \subset U$.
For a dynamical system $\{\varphi^t\}_{t \geq 0}$, a bounded closed invariant set K is
 \cite{LeonovKM-2015-EPJST}:
\begin{enumerate}[label=(\roman*)]
  \item a {\it (local)  attractor} if it is a minimal locally attractive set
        (i.e., $\lim_{t \to +\infty} \rho (K, \varphi^t(u)) = 0$ for all
        $u \in K(\varepsilon)$, where $K(\varepsilon)$ is a
        certain $\varepsilon$-neighborhood of set $K$),
  \item a {\it global attractor} if it is a minimal globally attractive set
        (i.e., $\lim_{t \to +\infty} \rho (K, \varphi^t(u)) = 0$
        for all $u \in \mathbb{R}^n$),
  \item a {\it (local)  B-attractor} if it is a minimal uniformly locally attractive set
        (i.e., for a certain $K(\varepsilon)$, any $\delta > 0$, and any bounded set $B$
        there exists $t(\delta, B) > 0$ such that
        $\varphi^t(B \cap K(\varepsilon)) \subset K(\delta)$ for all $t \geq t(\delta, B)$),
  \item a {\it global B-attractor} if it is a minimal uniformly globally attractive set
        (i.e., for any $\delta > 0$ and any bounded set $B \subset \mathbb{R}^n$
        there exists $t(\delta, B) > 0$ such that
        $\varphi^t(B) \subset K(\delta)$ for all $t \geq t(\delta, B)$).
\end{enumerate}


\noindent In the definition of attractor we assume closeness for the sake of uniqueness
since the closure of a locally attractive invariant set is also a locally
attractive invariant set (e.g., consider an attractor with excluded one of the
embedded unstable periodic orbits).
The above definition implies that a global attractor
involves the set of all equilibria.
The property of uniform attractivity
implies that a global B-attractor involves the unstable manifolds
of unstable equilibria
(the same is true for B-attractor if its neighborhood considered
contains some unstable equilibria).
If the dynamical system $\{\varphi^t\}_{t\geq0}$ possesses an absorbing set
$\mathcal{B}$, then the global attractor
can be constructed as follows: 
$\cap_{\tau > 0} \overline{\cup_{t \geq \tau} \varphi^t\left(\mathcal{B}\right)}$.

In the following, we consider system \eqref{sys:lorenz-general}
with two sets of parameters: $b = 1, \sigma = 4, A = 0.0052$,
and $r = 687.5$ or $r = 700$,
and visualize possible types of attractors
in Fig.~\ref{fig:selfexcite-all} and Fig.~\ref{fig:hidden-attr-all}, respectively.
Visualizations of chaotic self-excited ($r = 687.5$) and hidden ($r = 700$) attractors
in Fig.~\ref{fig:selfexcite-all} and Fig.~\ref{fig:hidden-attr-all}
are obtained by numerical integration of system \eqref{sys:lorenz-general}
on the time-interval $[0,60]$
with initial condition $P_1 = (10,60,800)$
and visualizations of numerical solutions after a transient process
(the separation of the trajectory into transition
process and approximation of attractor is rough).

\begin{figure}[th]
    \centering
    \subfloat[Local self-excited attractor.] {
      \label{fig:only-hidden}
      \includegraphics[width=0.31\textwidth]{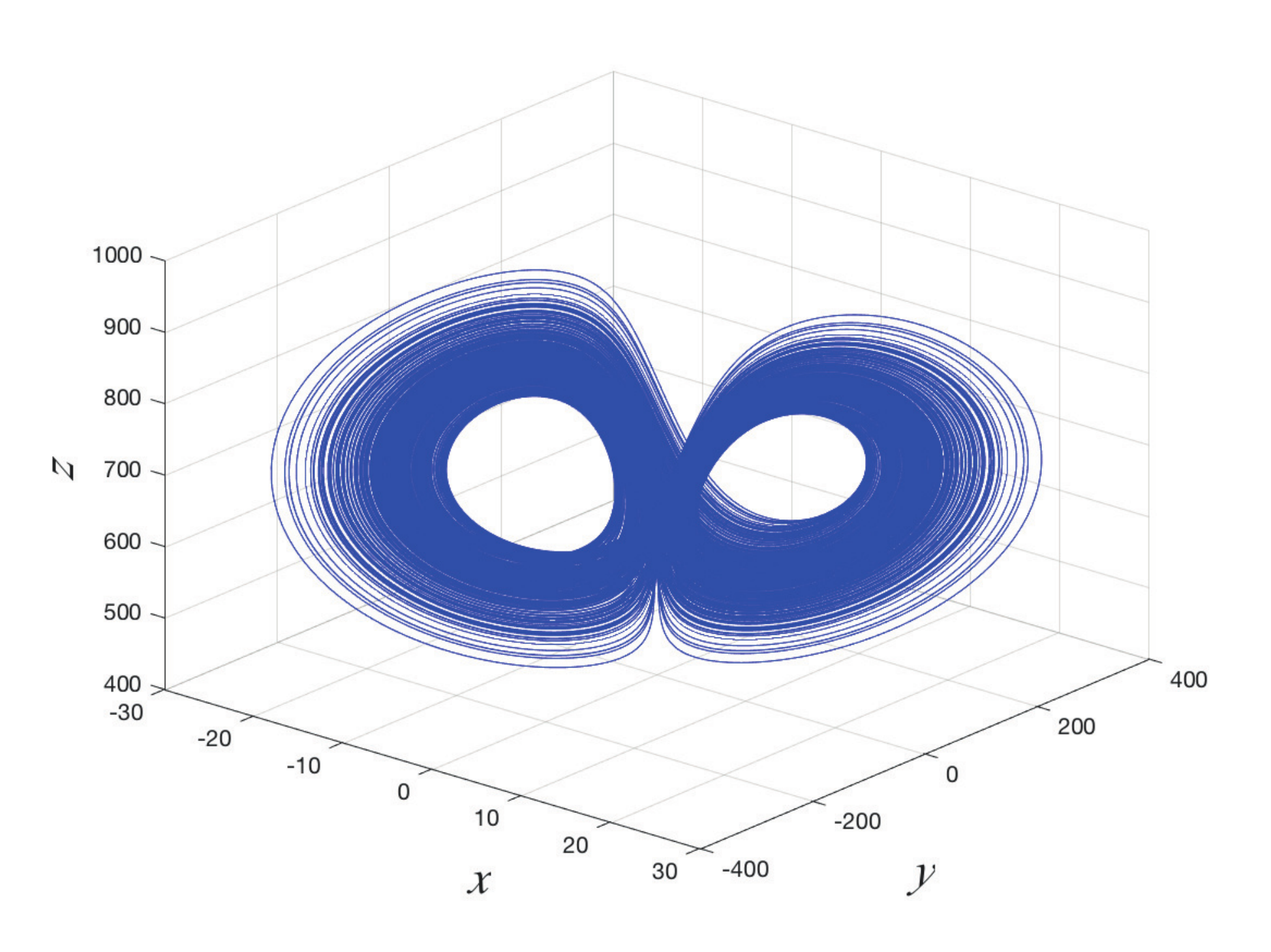}
    }
    \hfill
    \subfloat[Global attractor (the union of equilibria $S_{0,1,2}$ and local self-excited attractor).] {
      \label{fig:only-sepa}
      \includegraphics[width=0.31\textwidth]{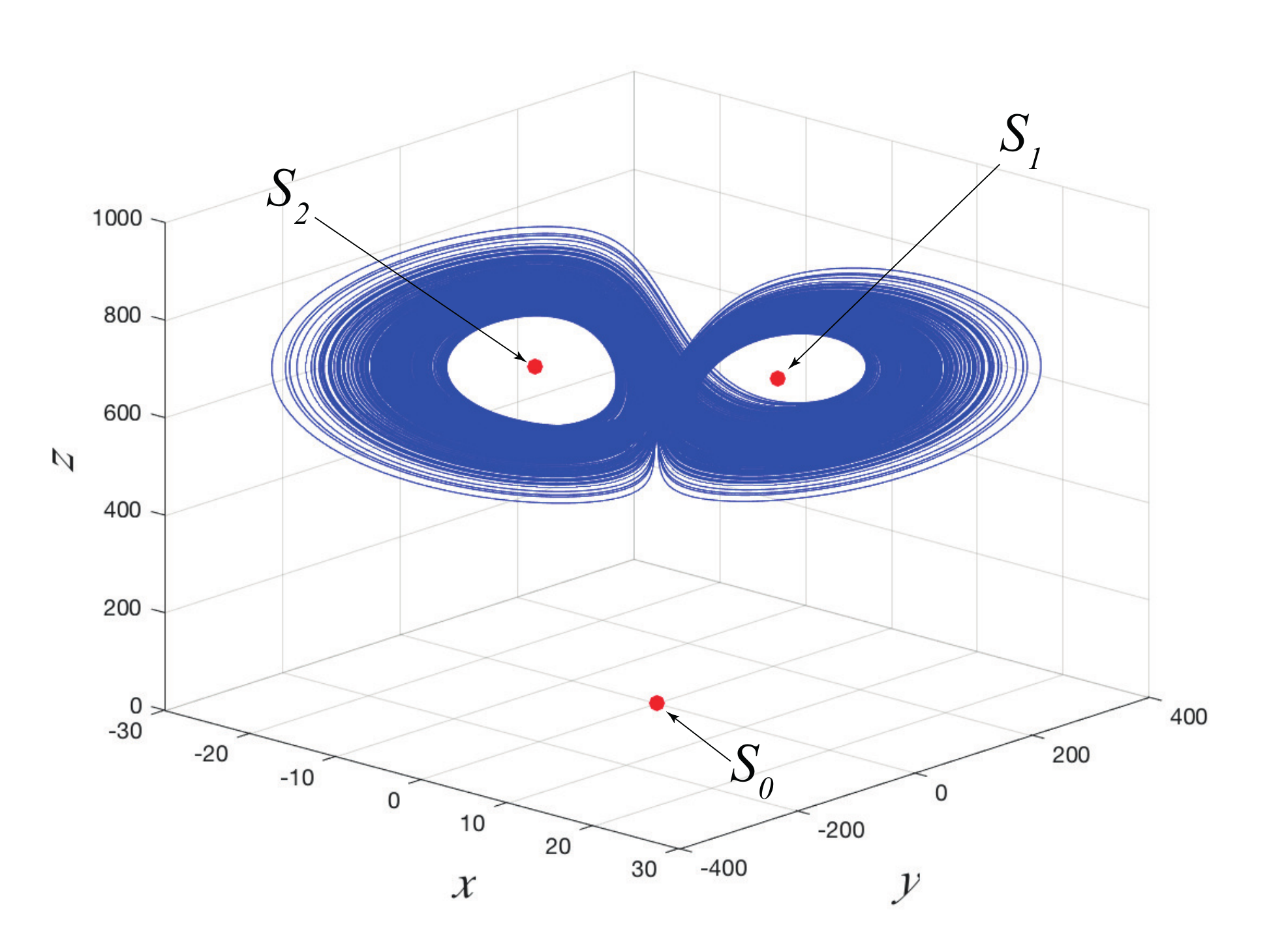}
    }
    \hfill
    \subfloat[Global B-attractor (includes local self-excited attractor (blue),
    equilibria $S_{0,1,2}$,
    and their unstable manifold (green)).] {
      \label{fig:glob}
      \includegraphics[width=0.31\textwidth]{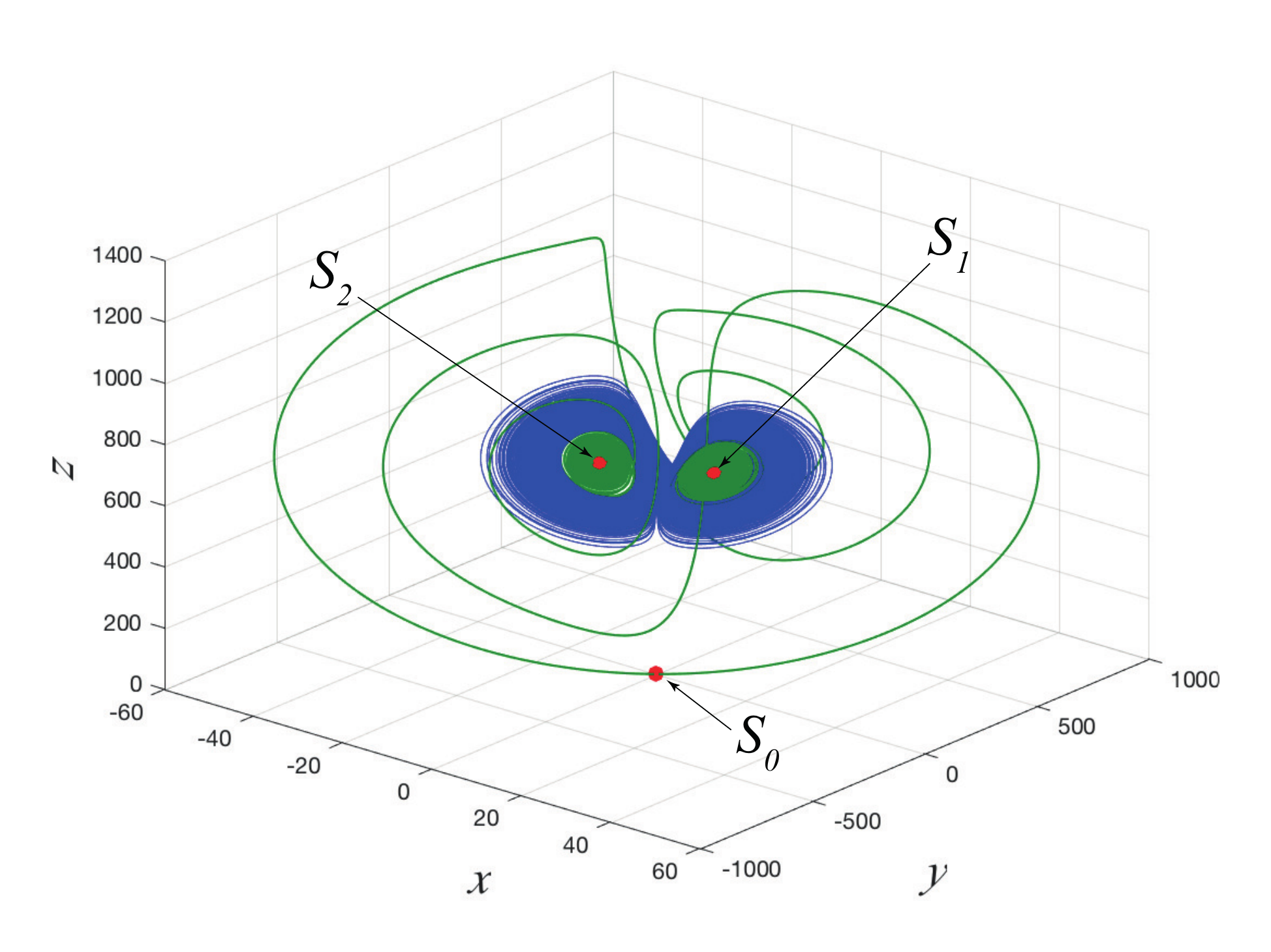}
    }
    \caption{Monostability. Numerical visualization of various types of attractors in system \eqref{sys:lorenz-general}
    with $b = 1$, $\sigma = 4$, $A = 0.0052$, $r=687.5$.}
  \label{fig:selfexcite-all}
\end{figure}
\begin{figure}[!ht]
    \centering
    \subfloat[Hidden attractor.] {
      \label{fig:only-hidden}
      \includegraphics[width=0.23\textwidth]{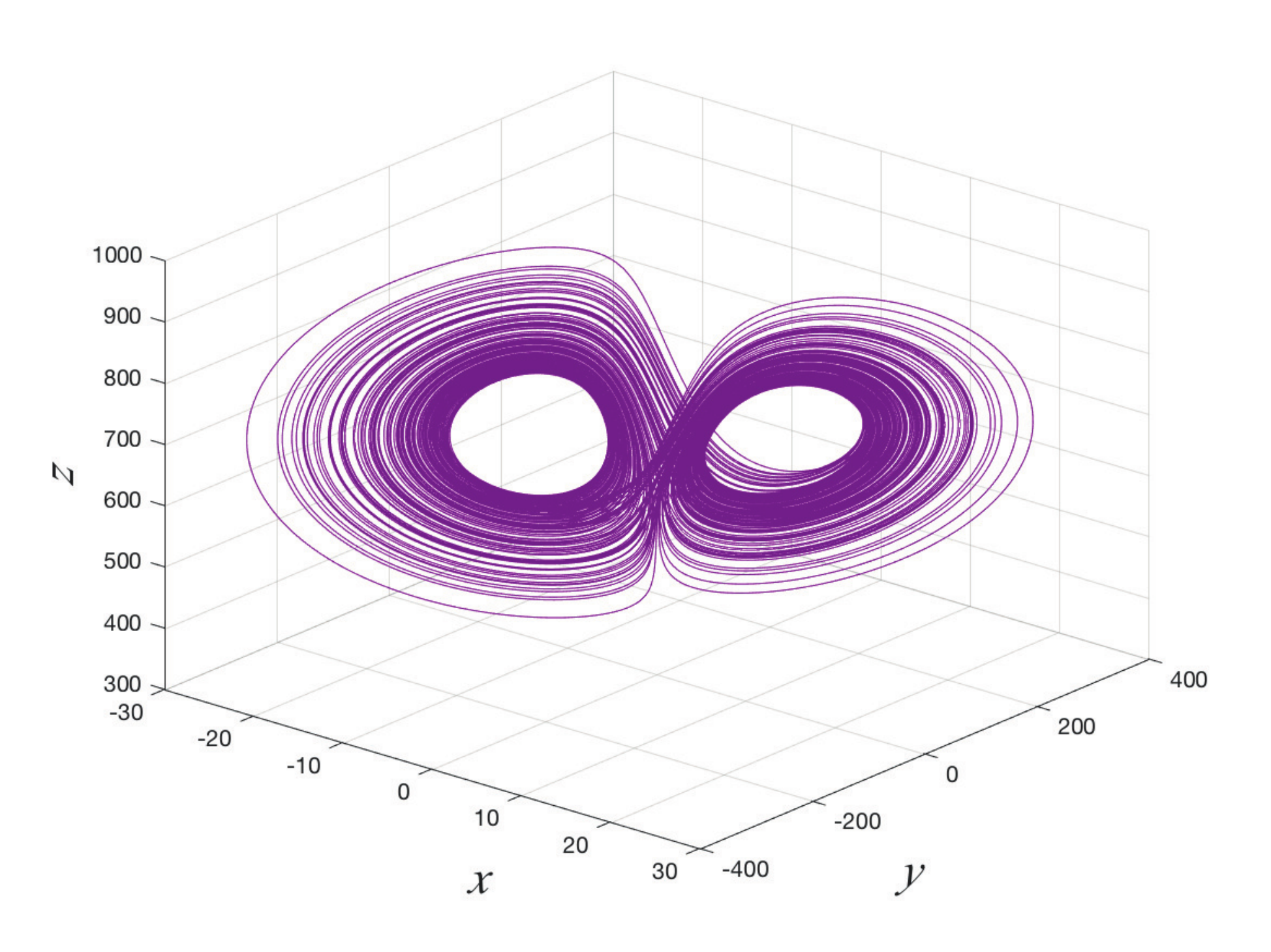}
    }
    \hfill
    \subfloat[B-attractor (includes equilibria $S_{0,1,2}$ and unstable manifold of $S_0$).] {
      \label{fig:only-sepa}
      \includegraphics[width=0.23\textwidth]{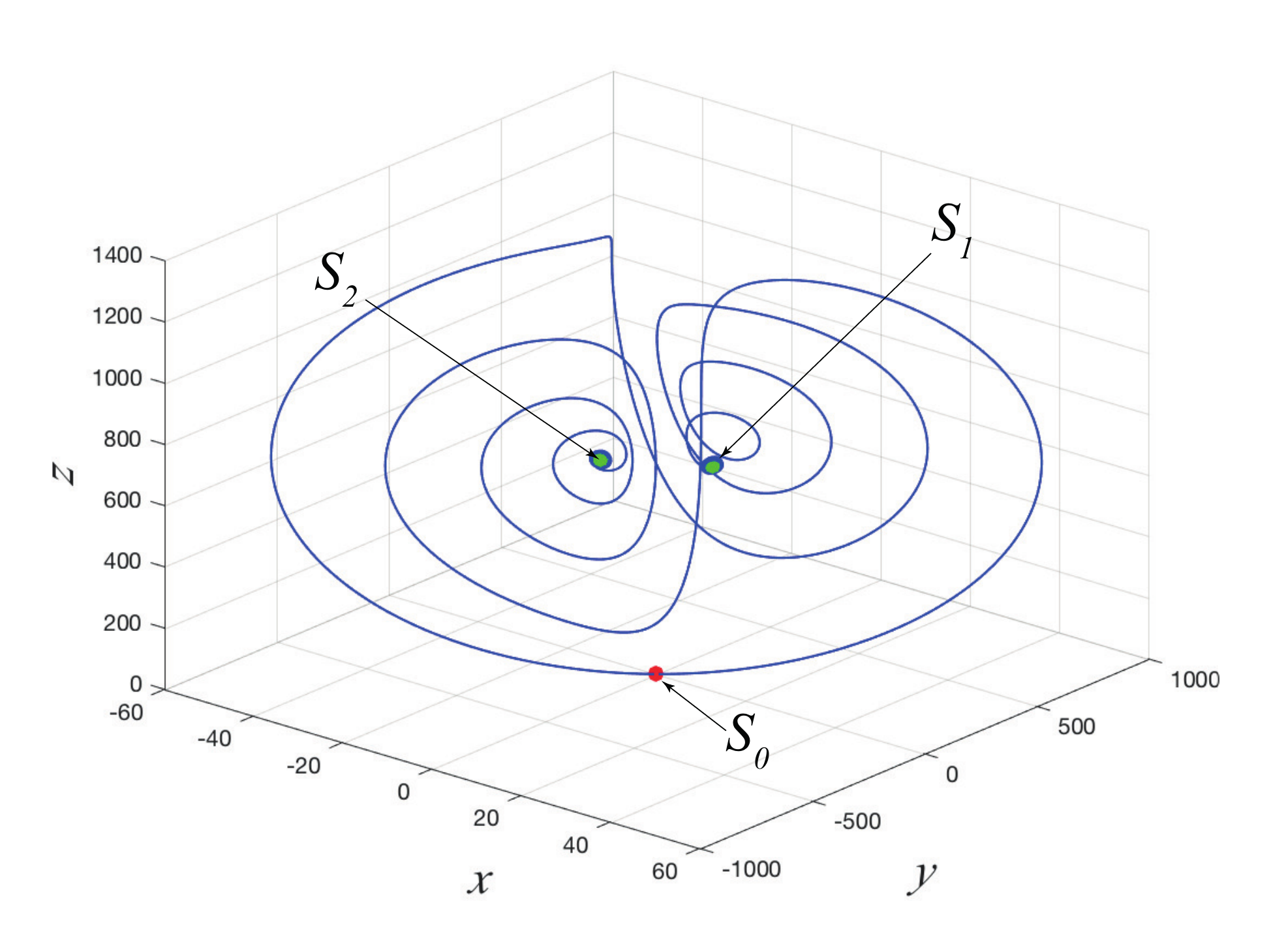}
    }
    \hfill
    \subfloat[Global attractor(the union of equilibria $S_{0,1,2}$ and local hidden attractor).] {
      \label{fig:glob}
      \includegraphics[width=0.23\textwidth]{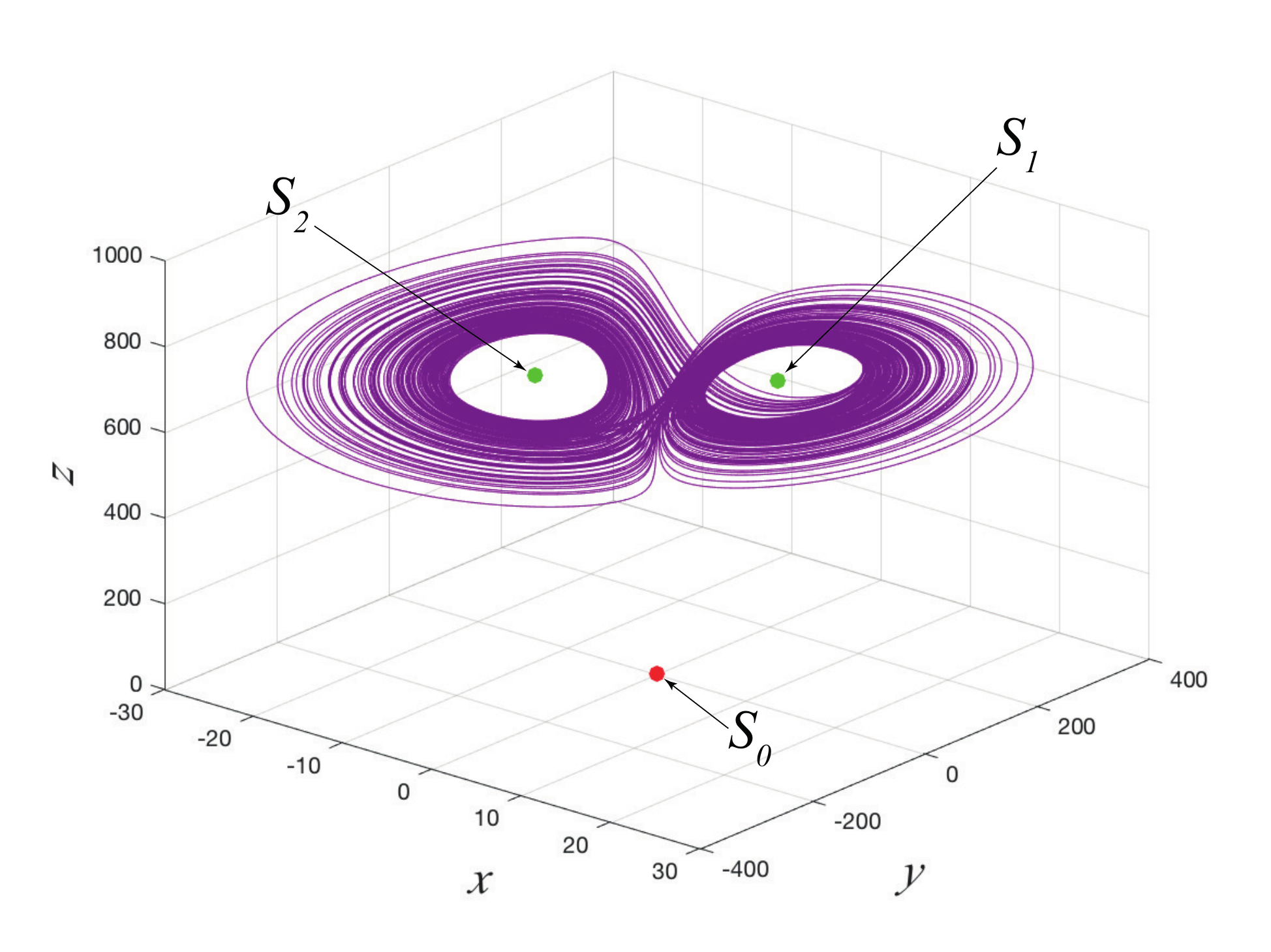}
    }
    \hfill
    \subfloat[Global B-attractor (the union of equilibria $S_{0,1,2}$, unstable manifold of $S_0$,
    and local hidden attractor).] 
    {
      \label{fig:glob}
      \includegraphics[width=0.23\textwidth]{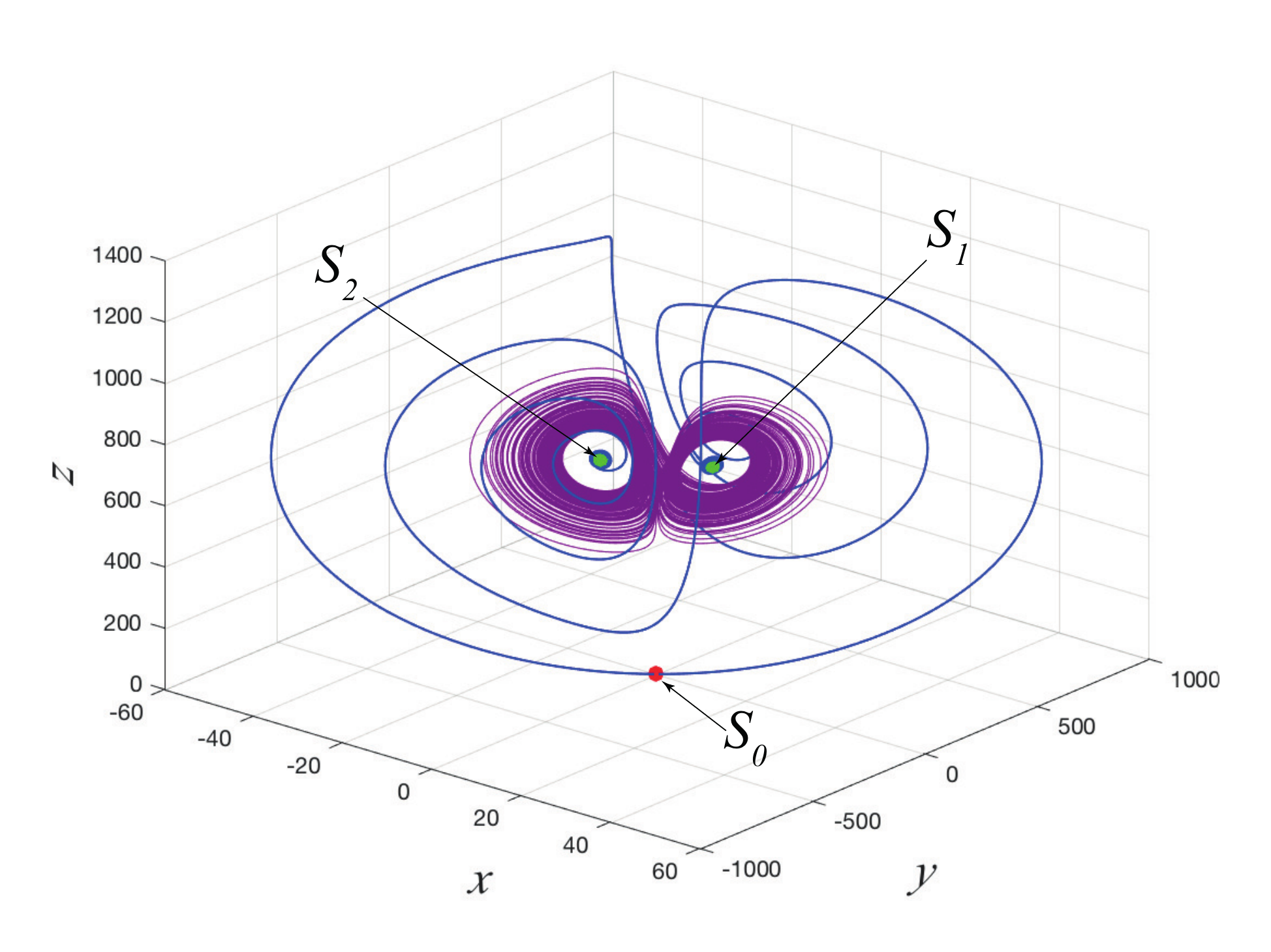}
    }
    \caption{Multistability. Numerical visualization of various types of attractors in system \eqref{sys:lorenz-general}
    with $b = 1$, $\sigma = 4$, $A = 0.0052$, $r=700$.}
  \label{fig:hidden-attr-all}
\end{figure}

Further we use the compact notations for
the finite-time Lyapunov dimensions:
$\dim_{\rm L}(t,u) = \dim_{\rm L}(\varphi^t_{\rm GD}, u)$,
$\dim_{\rm L}(t,K) = \dim_{\rm L}(\varphi^t_{\rm GD}, K)$,
and for the Lyapunov dimension:
$\dim_{\rm L}K = \dim_{\rm L}(\{\varphi^t_{\rm GD}\}_{t \geq0}, K)$.
For the chosen initial point $u_0=(x_0,y_0,z_0)$ and time interval $[0,T]$,
which are used to visualize the attractor $K$,
there are the following substantial questions related
to the computation of the finite-time Lyapunov dimension of $K$.
The first question is
whether there exists the limit
$\lim_{t \to +\infty}\dim_{\rm L}(\varphi^t,K)=\dim_{\rm L}K$
and, if not, whether for a given time interval $[0,T]$ the relation
$
 \dim_{\rm L}(T,K) \leq
 \inf_{t \in [0,T)}\dim_{\rm L}(t,K)
$ is true.
In general, there is no rigorous justification of the choice of $t$
and it is known that unexpected jumps
of $\dim_{\rm L}(t,K)$
can occur (see, e.g. Fig.~\ref{fig:tLD}).
Thus it is reasonable to compute $\inf_{t \in [0,T)}\dim_{\rm L}(T,K)$
instead of $\dim_{\rm L}(T,K)$,
but at the same time for any $T$
the value $\dim_{\rm L}(T,K)$ gives also a valid upper estimate for $\dim_{\rm H}K$.
The second question is whether
a given initial point $u_0$ belongs to the attractor $K$ or only to its basin of attraction
(and thus the whole semi-orbit $\{\varphi^t_{\rm GD}(u_0),\ t\geq0 \}$
belongs only to the basin of attraction),
and, if yes, whether $u_0$ is substantial for the Lyapunov dimension,
i.e. whether the relation
\(
  \dim_{\rm L}(K)
  =
  \dim_{\rm L} (u_0)
\)
is true or \(
  \dim_{\rm L}(K)
  =
  \dim_{\rm L} (K\setminus u_0)
\).
Since it is a challenging task to give
justified answers to these questions,
for numerical computation of the Lyapunov dimension
we have to consider a dense grid of points $K_{\rm grid}$
on a numerical approximation (visualization) of $K$
and approximate the Lyapunov dimension of attractor $K$
by $\max_{u \in K_{\rm grid}} \dim_{\rm L}(t, u)$.
Finally, in numerical experiments we can
expect
\[
\begin{aligned}
  & \dim_{\rm H}K
  \leq
  \dim_{\rm L}K
  \approx &
  \\
  &
  \qquad \qquad
  \approx \inf_{t \in [0,T]} \max_{u \in K_{\rm grid}} \dim_{\rm L}(t, u)
  = \max_{u \in K_{\rm grid}} \dim_{\rm L}(t_{\rm inf}, u)
  \leq
  \max_{u \in K_{\rm grid}} \dim_{\rm L}(T, u)
  \approx \dim_{\rm L}(T,K).
\end{aligned}
\]


\begin{figure}[h!]
    \centering
    \includegraphics[width=0.62\textwidth]{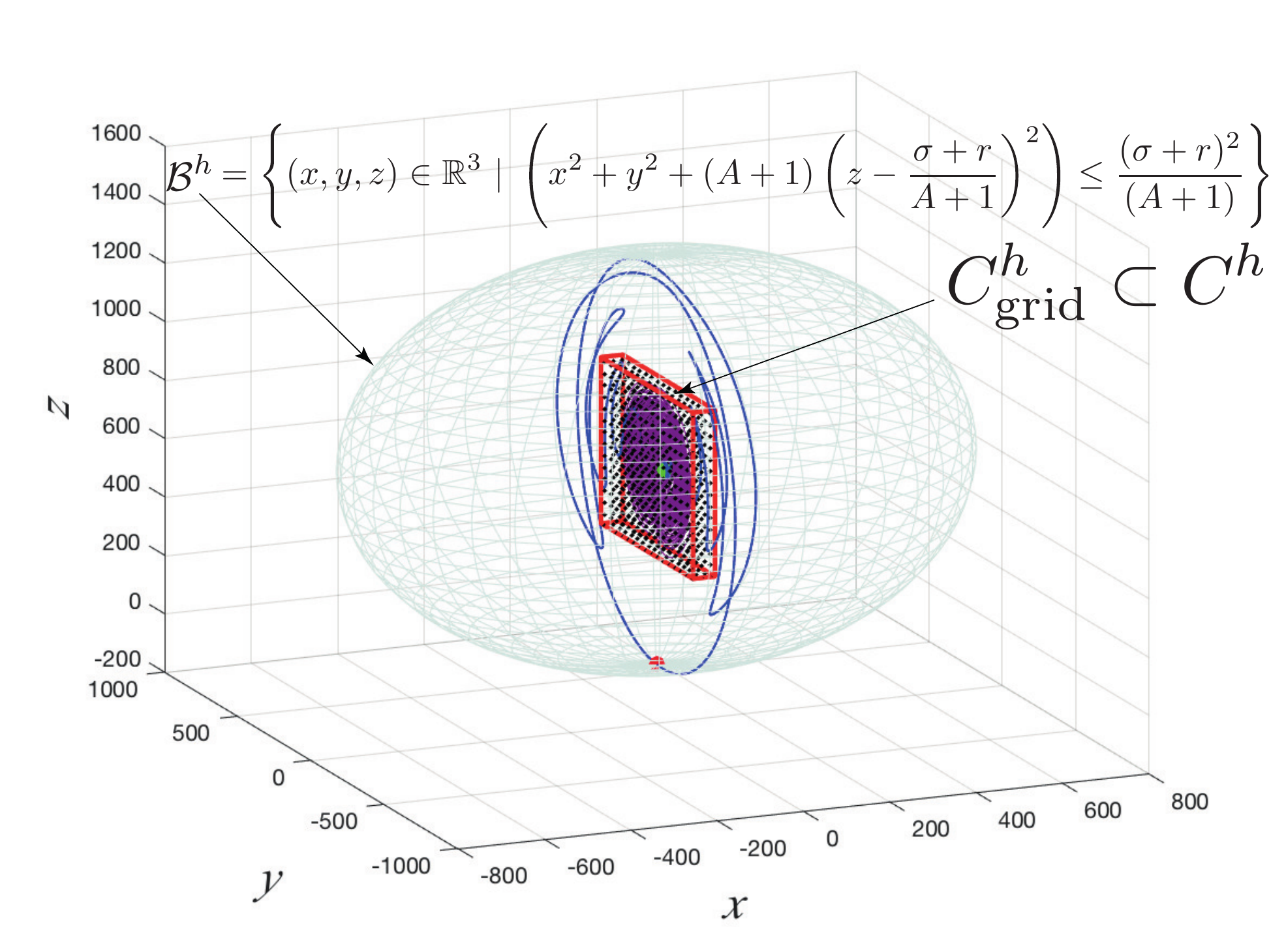}
    \caption{Localization of the hidden attractor by the absorbing set
    $\mathcal{B}^{h}$, cuboid
    $C^{h} =[-30,25]\times[-330,300]\times[395,956]$,
    and the corresponding grid of points $C^{h}_{\rm grid}$.
    }
    \label{fig:gd:grid}
\end{figure}

In Fig.~\ref{fig:gd:grid} is shown the grid of points $C^{h}_{\rm grid}$ covering the hidden attractor:
the grid points fill cuboid $C^{h} = [-30,25]\times[-330,300]\times[395,956]$
with the distance between points equals $5$;
the grid of points $C^{se}_{\rm grid}$ covering self-excited attractor
fill cuboid $C^{se}=[-25,25]\times[-305,290]\times[410,930]$).
The time interval is $[0,\, 60]$
and the integration method is MATLAB ode45.
Remark that if for a certain time the computed trajectory is out of the cuboid,
the corresponding value of finite-time local Lyapunov dimension
does not taken into account in the computation of maximum
of the finite-time local Lyapunov dimension
(there are trajectories with initial data in cuboid,
which are attracted to the zero equilibria, i.e. belong to its stable manifold,
e.g. system \eqref{sys:lorenz-general} for $x=y=0$ is $\dot z = -bz$).
The infimum on the time interval
is computed at the points $\{t_k\}_{1}^{N}$ with time step $t_{\Delta}=t_{i+1}-t_i=0.1$.
Note that if for a certain time $t=t_k$ the computed trajectory is out of the cuboid,
the corresponding value of finite-time local Lyapunov dimension
does not taken into account in the computation of maximum
of the finite-time local Lyapunov dimension
(there are trajectories with initial data in cuboid,
which are attracted to the zero equilibria, i.e. belong to its stable manifold,
e.g. system \eqref{sys:lorenz-general} for $x=y=0$ is $\dot z = -bz$).
For the finite-time Lyapunov exponents (FTLE) computation
we use MATLAB realization \cite{LeonovKM-2015-EPJST} of a method,
based on SVD decompositions.
For computation of the finite-time Lyapunov characteristic exponents (FTLCE)
we use MATLAB realization \cite{KuznetsovMV-2014-CNSNS} of a method,
based on QR decompositions.

For
both sets of parameters
(see Fig.~\ref{fig:selfexcite-all} and Fig.~\ref{fig:hidden-attr-all})
we compute:
1) finite-time local Lyapunov dimensions $\dim_{\rm L}(60,\cdot)$
at the point
$P_1 = (10,60,800)$, which belong to both grids $C^{h,se}_{\rm grid}$,
at the point $P_2 = (-0.0074, -0.0997, 0)$ on the unstable manifold of zero equilibria $S_0$;
2) maximums of the finite-time local Lyapunov dimensions at the points of grid
$\max_{u \in C^{h,se}_{\rm grid}} \dim_{\rm L}(t,u)$
for the time points $t=t_k=0.1\,k$ $(k=1,..,600)$ and
the infimum of the maximums;
3) the corresponding values, given by Kaplan-Yorke formula with respect to finite-time Lyapunov characteristic exponents.
The results are given in Table \ref{table:results:se} and \ref{table:results:hid}.

\begin{table}[!h]
\centering
\caption{The set of parameters corresponding the self-excited attractor (see Fig.~\ref{fig:selfexcite-all})}
\begin{tabular}{
|>{\centering}m{3.4cm}<{\centering}||
>{\centering}m{3.3cm}<{\centering}|
>{\centering}m{5cm}<{\centering}|
>{\centering}m{2cm}<{\centering}|
>{\centering}m{2.3cm}<{\centering}|}
\hline
& $t = 60$ \\ $u = (10, 60, 800)$ & $t = 60$ \\ $u = (-0.0074, -0.0997, 0)$ &
$t = 60$ \\ $\max_{u \in C^{se}_{\rm grid}}$ &
$\displaystyle \inf_{t \in [0,\, 60]} \max_{u \in C^{se}_{\rm grid}}$
\tabularnewline\hline
$\dim_{\rm L}(t, \, u)$ \\ (SVD)& $2.1345$ & $2.1468$ & $2.1721$ & $2.1699$
\tabularnewline\hline
${\scriptstyle {\rm d}_{\rm L}^{\rm KY}(\{{\rm LCE}_i(t, \, u)\}_{i=1}^3)}$ \\ (QR) &
$2.1414$ & $2.1262$ & $2.1876$ & $2.1854$
\tabularnewline\hline
\end{tabular}
\label{table:results:se}
\end{table}

\begin{table}[!h]
\centering
\caption{The set of parameters corresponding the hidden attractor (see Fig.~\ref{fig:hidden-attr-all})}
\begin{tabular}{
|>{\centering}m{3.4cm}<{\centering}||
>{\centering}m{3.3cm}<{\centering}|
>{\centering}m{5cm}<{\centering}|
>{\centering}m{2cm}<{\centering}|
>{\centering}m{2.3cm}<{\centering}|}
\hline
& $t = 60$ \\ $u = (10, 60, 800)$ & $t = 60$ \\ $u = (-0.0074, -0.0997, 0)$ &
$t = 60$ \\ $\max_{u \in C^{h}_{\rm grid}}$ & $\displaystyle \inf_{t \in [0,\, 60]} \max_{u \in C^{h}_{\rm grid}}$
\tabularnewline\hline
$\dim_{\rm L}(t, \, u)$ \\ (SVD) & $2.1271$ & $1.2335$ & $2.1638$ & $2.1580$
\tabularnewline\hline
${\scriptstyle {\rm d}_{\rm L}^{\rm KY}(\{{\rm LCE}_i(t, \, u)\}_{i=1}^3)}$ \\ (QR) &
$2.1372$ & $1.0139$ & $2.1745$ & $2.1735$
\tabularnewline\hline
\end{tabular}
\label{table:results:hid}
\end{table}

\begin{figure}[h!]
  \centering
  \subfloat[Approximation via SVD based method.]{
    \label{fig:tLD:se:SVD}
    \includegraphics[width=\textwidth]{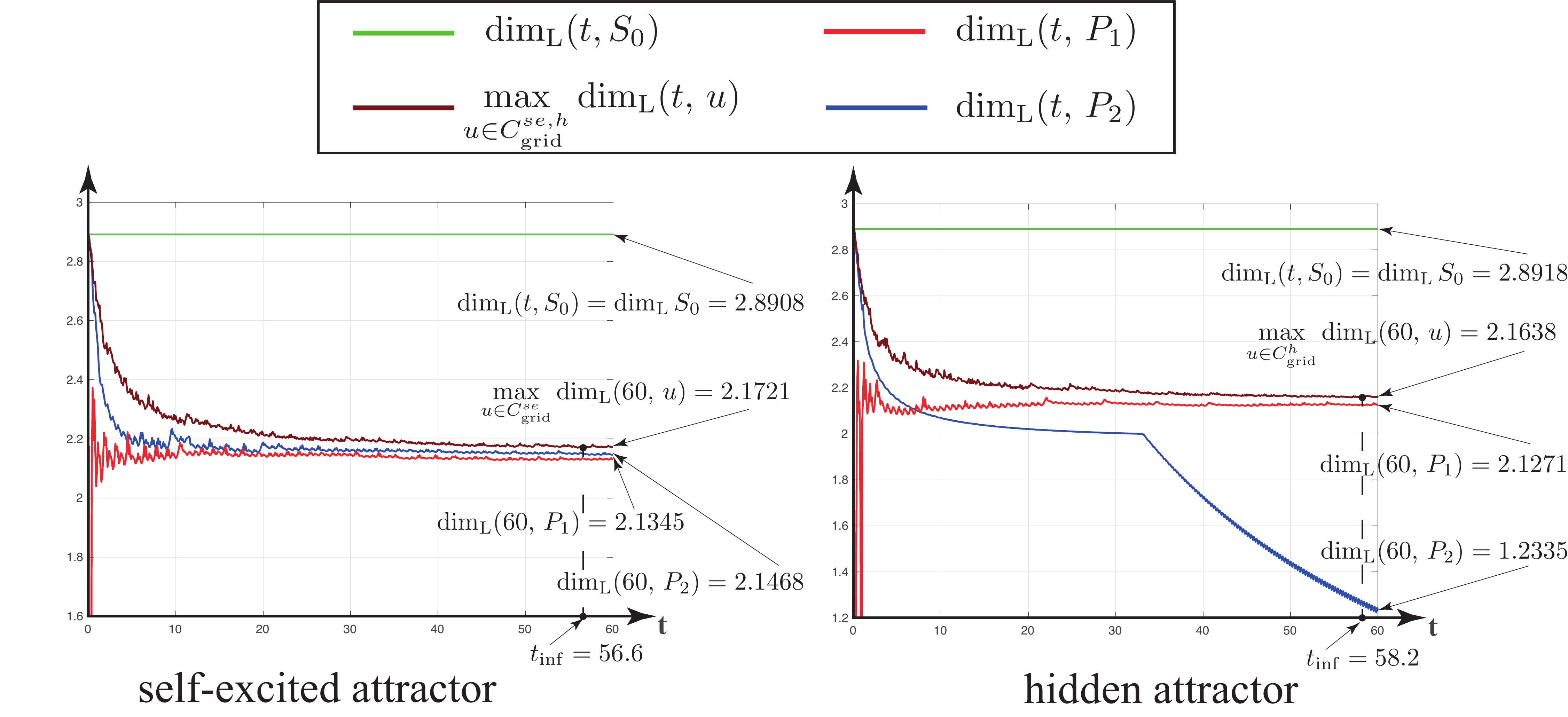}
  }

  \subfloat[Approximation via QR based method.]{
    \label{fig:tLD:QR}
    \includegraphics[width=\textwidth]{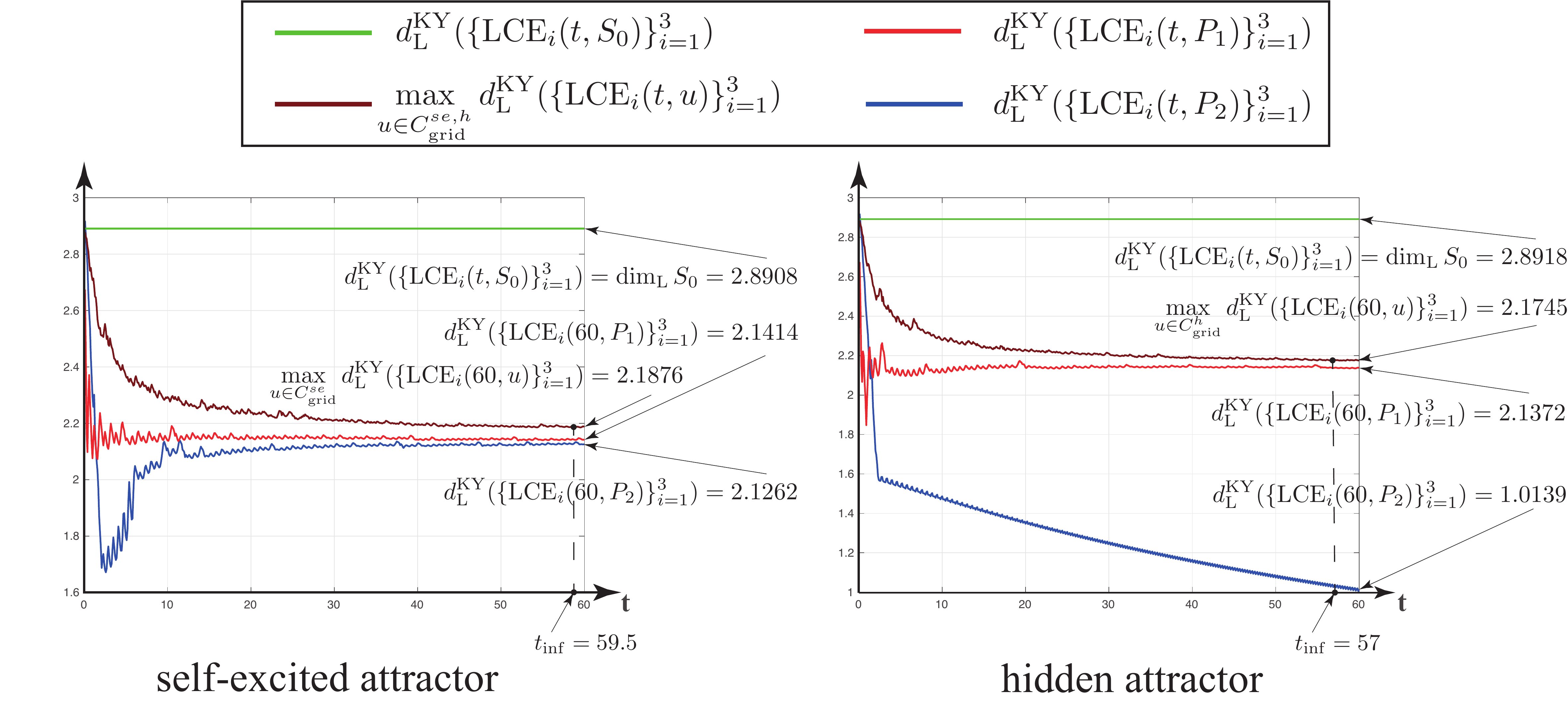}
  }
  \caption{
    Dynamics of the finite-time local Lyapunov dimensions
    on the time interval $t \in \lbrack 0, 60 \rbrack$:
    maximum on the grid of points (dark red),
    at the point $P_1 = (10, \, 60, \, 800) \in C^{se,h}_{\rm grid}$ (light red),
    at $P_2 = (-0.0074, \,  -0.0997, \, 0)$
    from 1D unstable manifold of $S_0$ (blue);
    at the equilibrium $S_0$ (green).
  }
  \label{fig:tLD}
\end{figure}

The behavior of finite-time local Lyapunov dimensions
for different points and their maximum on a grid of points
is shown in Fig.~\ref{fig:tLD}.
These values are in good correspondence with the exact Lyapunov dimension,
obtained in Corollary~\ref{conseq}.
For the global attractor $K_{\text{global}}$
and global B-attractor $K_{\text{global B}}$ in Fig.~\ref{fig:selfexcite-all}(b,c)
we have
\[
  \dim_{\rm L} K_{\text{global}} = \dim_{\rm L} K_{\text{global B}}
  = 2.8908...
\]
Since the global B-attractor $K_{\text{global B}}$ in Fig.~\ref{fig:selfexcite-all}(c)
involves two-dimensional unstable manifolds of equilibria $S_{1,2}$,
we have
\[
  2 \leq \dim_{\rm H} K_{\text{global B}} \leq \dim_{\rm L} K_{\text{global B}}.
\]

For the B-attractor $K_{\text{B}}$, the global attractor $K_{\text{global}}$,
global B-attractor $K_{\text{global B}}$,
in Fig.~\ref{fig:hidden-attr-all}(b,c,d)
we have
\[
  \dim_{\rm L} K_{\text{global}} =
  \dim_{\rm L} K_{\text{global B}} =
  \dim_{\rm L} K_{\text{B}}
  = 2.8918...
\]
Since the global B-attractor $K_{\text{global B}}$ in Fig.~\ref{fig:hidden-attr-all}(d)
involves one-dimensional unstable manifolds of equilibrium $S_{0}$,
we have
\[
  1 \leq \dim_{\rm H} K_{\text{global B}} \leq \dim_{\rm L} K_{\text{global B}}.
\]

Remark that
the absorbing sets $\mathcal{B}^{se}=\mathcal{B}(687.5,4,0.0052)$
and $\mathcal{B}^{h}=\mathcal{B}(700,4,0.0052)$
involve all the considered attractors in
Fig.~\ref{fig:selfexcite-all} and Fig.~\ref{fig:hidden-attr-all}, respectively.
Thus, for the corresponding grid of points
by estimation \eqref{dLKYeig} with $S=I$,
we get an estimate for any attractor $K$ in Fig.~\ref{fig:selfexcite-all}
(here the distance between grid points is 20):
\begin{equation}\label{KYeigonabs-se}
 \dim_{\rm H}K \leq
 \dim_{\rm L} K \leq
 \sup_{u \in \mathcal{B}^{se}}d_{\rm L}^{\rm KY}\big(\{\lambda_{j}(u)\}_{i=1}^n\big)
 \approx
 \sup_{u \in \mathcal{B}^{se}_{\rm grid}}d_{\rm L}^{\rm KY}\big(\{\lambda_{j}(u)\}_{i=1}^n\big)
 = 2.982747... \, ,
\end{equation}
and for any attractor $K$ in Fig.~\ref{fig:hidden-attr-all}
\begin{equation}\label{KYeigonabs-h}
 \dim_{\rm H}K \leq
 \dim_{\rm L} K \leq
 \sup_{u \in \mathcal{B}^{h}}d_{\rm L}^{\rm KY}\big(\{\lambda_{j}(u)\}_{i=1}^n\big)
 \approx
 \sup_{u \in \mathcal{B}^{h}_{\rm grid}}d_{\rm L}^{\rm KY}\big(\{\lambda_{j}(u)\}_{i=1}^n\big)
 = 2.983037... \, .
\end{equation}

\begin{figure}[h!]
  \centering
    \includegraphics[width=\textwidth]{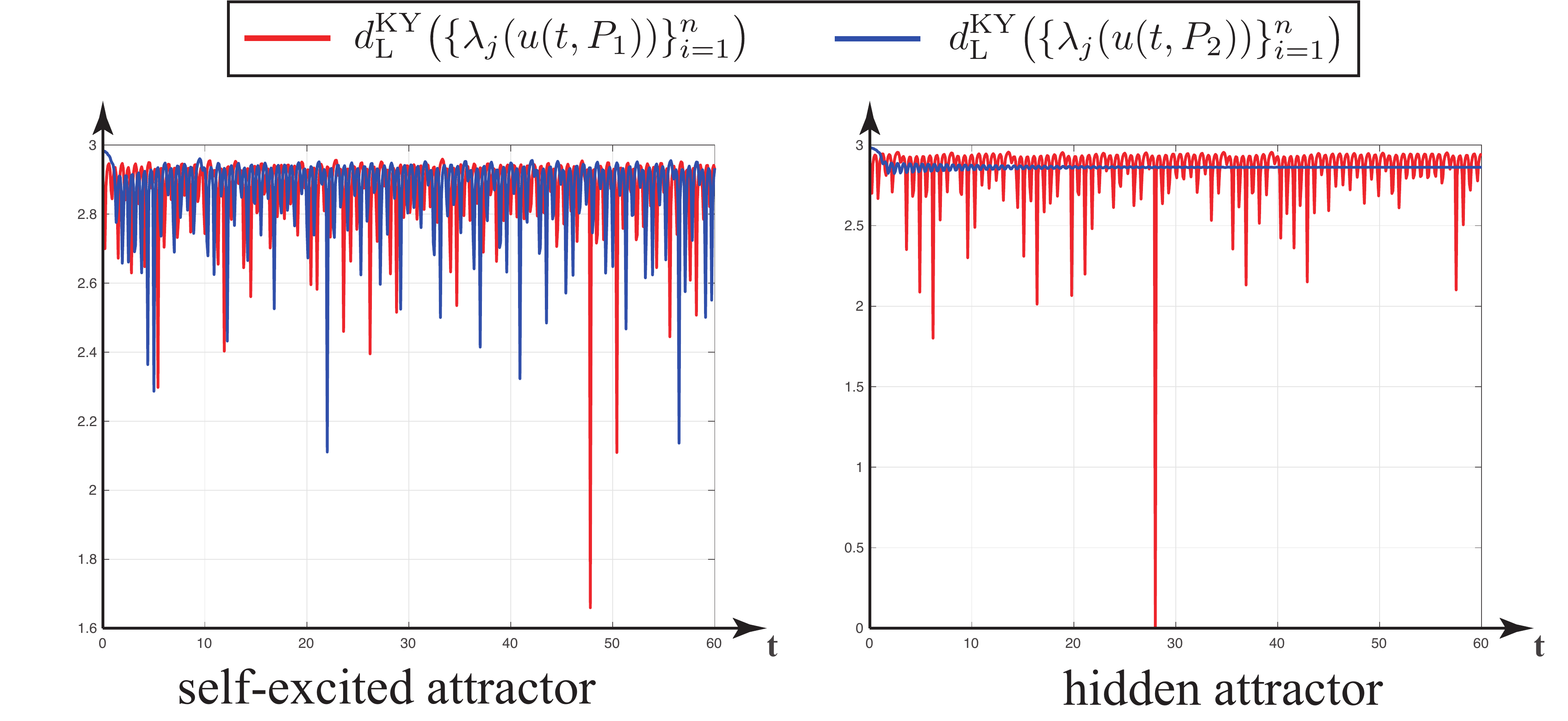}
  \caption{
    Dynamics of $d_{\rm L}^{\rm KY}\big(\{\lambda_{j}(u(t,P_1))\}_{i=1}^n\big)$ (red) and
    $d_{\rm L}^{\rm KY}\big(\{\lambda_{j}(u(t,P_2))\}_{i=1}^n\big)$ (blue)
    along the trajectories with initial data
    $P_1 = (10, \, 60, \, 800) \in C^{se,h}_{\rm grid}$ and
    $P_2 = (-0.0074, \,  -0.0997, \, 0)$.
  }
  \label{fig:tLambda}
\end{figure}

The above numerical experiments lead to the following important conluding remarks.
While the Lyapunov dimension,
unlike the Hausdorff dimension,
is not a dimension in the rigorous sense \cite{HurewiczW-1941,Kuratowski-1966}
(e.g. the Lyapunov dimension of a saddle point or a periodic orbit
can be noninteger and has different values including those close to $n$),
it gives an upper estimate of the Hausdorff dimension.
The sets with noninterger Hausdorff dimension are referred as the fractal sets \cite{EckmannR-1985}.
Let the attractor $K$ or the corresponding absorbing set $\mathcal{B} \supset K$
is known (see, e.g. \eqref{absorb_set}).
If one's purpose is to demonstrate that $\dim_{\rm H} K \leq \dim_{\rm L} K < n$,
then it can be achieved by Proposition~\ref{thm:dLKYeig}
without integration of the considered dynamical system
(see, e.g. \eqref{KYeigonabs-se} and \eqref{KYeigonabs-h}).
In general, one can not expect to get the same values of
the symmetrized Jacobian matrix eigenvalues
at different points (see, e.g. Fig.~\ref{fig:tLambda}),
thus the maximum of $d_{\rm L}^{\rm KY}\big(\{\lambda_{j}(u)\}_{i=1}^n\big)$
on a grid of points has to be considered.
If the purpose is to get a precise estimation of the Hausdorff dimension,
then one can use \eqref{HDOinf} and have to compute the Lyapunov dimension.
To be able to repeat a computation of Lyapunov dimension one need to know
the considered compact invariant set $K$ and initial points of considered trajectories
$\{u_i\}_{i=1}^{N}=K_{\rm grid}$ on the set,
time interval $(0,T] = \bigcup_{i=0}^{M-1} (t_i,t_{i+1}]$,
and the method of Lyapunov exponents computation.
Finite-time Lyapunov dimension is defined by singular values functions
and for computing the corresponding Lyapunov exponents
one has to use a numerical algorithms based on the SVD decomposition
(not on the QR decomposition).
While unexpected jumps in the values of the local Lyapunov exponents and Lyapunov dimension
may occur, there is no rigorous justification of the choice of time $T$, and,
thus, the $\inf_{t \in \{t_i\}_{i=0}^{M}, t_M=T}$ of the finite-time Lyapunov dimensions
often gives better estimates.
In general, in numerical experiments one can not expect to get the same values of
the finite-time local Lyapunov exponents and Lyapunov dimension for different points,
thus the maximum of the finite-time local Lyapunov dimensions
on the grid of point ($\max_{u \in K_{\rm grid}}$) has to be considered.
Remark that in the work \cite[p.190]{FredericksonKYY-1983}
Kaplan and Yorke called the limit values of the finite-time Lyapunov exponents
$\{\lim\limits_{t\to+\infty}\LEs_i(t,u)\}_{i=1}^n = \{\LEs_i(u)\}_{i=1}^n$,
if they exist and are the same for all $u \in K$
(and therefore \cite{KuznetsovL-2016-ArXiv}
$\dim_{\rm L} K = d_{\rm L}^{\rm KY}(\{\LEs_i(u_0)\}_{i=1}^n)$ for any $u_0 \in K$),
the \emph{absolute} ones and wrote that such absolute values \emph{rarely exist}.

\section{Conclusion and further steps} 
\label{sec:conclusion}
In this paper the Lyapunov dimension of attractors in the Glukhovsky-Dolzhansky fluid convection model has been studied by analytical and numerical methods.
In studying we follow a rigorous approach to the definition
of the Lyapunov dimension and justification of its computation by the Kaplan-Yorke formula,
without using statistical physics assumptions.
The exact Lyapunov dimension formula for the global attractors
is obtained and peculiarities
of the Lyapunov dimension estimation for self-excited and hidden attractors
are discussed.
A tutorial on numerical estimation of the Lyapunov dimension
on the example of the Glukhovsky-Dolzhansky model is presented.

\section*{Acknowledgements}
Authors would like to thank Alexander Gluhovsky,
Professor of the Department of Earth \& Atmospheric Sciences and
Department of Statistics, Purdue University, USA,
for the fruitful discussion and valuable comments
on the low-order models.
This work was supported by the Russian Science Foundation (14-21-00041).


\begin{thebibliography}{10}
\expandafter\ifx\csname url\endcsname\relax
  \def\url#1{\texttt{#1}}\fi
\expandafter\ifx\csname urlprefix\endcsname\relax\def\urlprefix{URL }\fi
\expandafter\ifx\csname href\endcsname\relax
  \def\href#1#2{#2} \def\path#1{#1}\fi

\bibitem{Thompson-1961}
P.~Thompson, Numerical weather analysis and prediction, Macmillan New York,
  1961.

\bibitem{Obukhov-1973}
A.~Obukhov, On the problem of nonlinear interactions in fluid dynamics,
  Gerlands Beitraege zur Geophysik 82~(4) (1973) 282--290.

\bibitem{Glukhovsky-1982}
A.~B. Glukhovsky, Nonlinear systems that are superpositions of gyrostats, Sov.
  Phys. Dokl. 27~(10) (1982) 823--825.

\bibitem{GluhovskyG-2016}
A.~Gluhovsky, K.~Grady, Effective low-order models for atmospheric dynamics and
  time series analysis, Chaos: An Interdisciplinary Journal of Nonlinear
  Science 26~(2) (2016) 023119.

\bibitem{Lorenz-1963}
E.~N. Lorenz, Deterministic nonperiodic flow, J. Atmos. Sci. 20~(2) (1963)
  130--141.

\bibitem{Vallis-1986}
G.~K. Vallis, El {N}i{\~n}o: {A} chaotic dynamical system?, Science 232~(4747)
  (1986) 243--245.

\bibitem{GlukhovskyD-1980}
A.~B. Glukhovsky, F.~V. Dolzhansky, Three component models of convection in a
  rotating fluid, Izv. Acad. Sci. USSR, Atmos. Oceanic Phys. 16 (1980)
  311--318.

\bibitem{Tucker-1999}
W.~Tucker, The {L}orenz attractor exists, Comptes Rendus de l'Academie des
  Sciences - Series I - Mathematics 328~(12) (1999) 1197 -- 1202.

\bibitem{Stewart-2000}
I.~Stewart, Mathematics: The {L}orenz attractor exists, Nature 406~(6799)
  (2000) 948--949.

\bibitem{MenckHMK-2013}
P.~J. Menck, J.~Heitzig, N.~Marwan, J.~Kurths, How basin stability complements
  the linear-stability paradigm, Nature Physics 9~(2) (2013) 89--92.

\bibitem{PisarchikF-2014}
A.~Pisarchik, U.~Feudel, Control of multistability, Physics Reports 540~(4)
  (2014) 167–218.
\newblock \href {http://dx.doi.org/10.1016/j.physrep.2014.02.007}
  {\path{doi:10.1016/j.physrep.2014.02.007}}.

\bibitem{DudkowskiJKKLP-2016}
D.~Dudkowski, S.~Jafari, T.~Kapitaniak, N.~Kuznetsov, G.~Leonov, A.~Prasad,
  Hidden attractors in dynamical systems, Physics Reports 637 (2016) 1--50.
\newblock \href {http://dx.doi.org/10.1016/j.physrep.2016.05.002}
  {\path{doi:10.1016/j.physrep.2016.05.002}}.

\bibitem{GrassbergerP-1983}
P.~Grassberger, I.~Procaccia, Measuring the strangeness of strange attractors,
  Physica D: Nonlinear Phenomena 9~(1-2) (1983) 189--208.

\bibitem{Rabinovich-1978}
M.~I. Rabinovich, Stochastic autooscillations and turbulence, Uspehi
  Physicheskih Nauk 125~(1) (1978) 123--168.

\bibitem{LeonovB-1992}
G.~A. Leonov, V.~A. Boichenko, Lyapunov's direct method in the estimation of
  the {H}ausdorff dimension of attractors, Acta Applicandae Mathematicae 26~(1)
  (1992) 1--60.

\bibitem{PikovskiRT-1978}
A.~S. Pikovski, M.~I. Rabinovich, V.~Y. Trakhtengerts, Onset of stochasticity
  in decay confinement of parametric instability, Sov. Phys. JETP 47 (1978)
  715--719.

\bibitem{KuznetsovLMS-2016-INCAAM}
N.~Kuznetsov, G.~Leonov, T.~Mokaev, S.~Seledzhi, Hidden attractor in the
  {R}abinovich system, {C}hua circuits and {PLL}, AIP Conference Proceedings
  1738~(1), art. num. 210008.

\bibitem{LeonovKV-2011-PLA}
G.~Leonov, N.~Kuznetsov, V.~Vagaitsev, Localization of hidden {C}hua's
  attractors, Physics Letters A 375~(23) (2011) 2230--2233.
\newblock \href {http://dx.doi.org/10.1016/j.physleta.2011.04.037}
  {\path{doi:10.1016/j.physleta.2011.04.037}}.

\bibitem{LeonovKV-2012-PhysD}
G.~Leonov, N.~Kuznetsov, V.~Vagaitsev, Hidden attractor in smooth {C}hua
  systems, Physica D: Nonlinear Phenomena 241~(18) (2012) 1482--1486.
\newblock \href {http://dx.doi.org/10.1016/j.physd.2012.05.016}
  {\path{doi:10.1016/j.physd.2012.05.016}}.

\bibitem{LeonovK-2013-IJBC}
G.~Leonov, N.~Kuznetsov, Hidden attractors in dynamical systems. {F}rom hidden
  oscillations in {H}ilbert-{K}olmogorov, {A}izerman, and {K}alman problems to
  hidden chaotic attractors in {C}hua circuits, International Journal of
  Bifurcation and Chaos 23~(1), {a}rt. no. 1330002.
\newblock \href {http://dx.doi.org/10.1142/S0218127413300024}
  {\path{doi:10.1142/S0218127413300024}}.

\bibitem{LeonovKM-2015-EPJST}
G.~Leonov, N.~Kuznetsov, T.~Mokaev, Homoclinic orbits, and self-excited and
  hidden attractors in a {L}orenz-like system describing convective fluid
  motion, Eur. Phys. J. Special Topics 224~(8) (2015) 1421--1458.
\newblock \href {http://dx.doi.org/10.1140/epjst/e2015-02470-3}
  {\path{doi:10.1140/epjst/e2015-02470-3}}.

\bibitem{AndronovVKh-1966}
A.~A. Andronov, E.~A. Vitt, S.~E. Khaikin, Theory of Oscillators, Pergamon
  Press, Oxford, 1966.

\bibitem{Jenkins-2013}
A.~Jenkins, Self-oscillation, Physics Reports 525~(2) (2013) 167--222.

\bibitem{KuznetsovLV-2010-IFAC}
N.~Kuznetsov, G.~Leonov, V.~Vagaitsev, Analytical-numerical method for
  attractor localization of generalized {C}hua's system, IFAC Proceedings
  Volumes (IFAC-PapersOnline) 4~(1) (2010) 29--33.
\newblock \href {http://dx.doi.org/10.3182/20100826-3-TR-4016.00009}
  {\path{doi:10.3182/20100826-3-TR-4016.00009}}.

\bibitem{KuznetsovKLV-2013}
N.~Kuznetsov, O.~Kuznetsova, G.~Leonov, V.~Vagaitsev, Analytical-numerical
  localization of hidden attractor in electrical {C}hua's circuit, Informatics
  in Control, Automation and Robotics, Lecture Notes in Electrical Engineering,
  Volume 174, Part 4 174~(4) (2013) 149--158.
\newblock \href {http://dx.doi.org/10.1007/978-3-642-31353-0\_11}
  {\path{doi:10.1007/978-3-642-31353-0\_11}}.

\bibitem{Kuznetsov-2016}
N.~Kuznetsov, Hidden attractors in fundamental problems and engineering models.
  {A} short survey., Lecture Notes in Electrical Engineering 371 (2016) 13--25,
  (Plenary lecture at AETA 2015: Recent Advances in Electrical Engineering and
  Related Sciences).
\newblock \href {http://dx.doi.org/10.1007/978-3-319-27247-4\_2}
  {\path{doi:10.1007/978-3-319-27247-4\_2}}.

\bibitem{ShahzadPAJH-2015-HA}
M.~Shahzad, V.-T. Pham, M.~Ahmad, S.~Jafari, F.~Hadaeghi, Synchronization and
  circuit design of a chaotic system with coexisting hidden attractors,
  European Physical Journal: Special Topics 224~(8) (2015) 1637--1652.

\bibitem{BrezetskyiDK-2015-HA}
S.~Brezetskyi, D.~Dudkowski, T.~Kapitaniak, Rare and hidden attractors in van
  der {P}ol-{D}uffing oscillators, European Physical Journal: Special Topics
  224~(8) (2015) 1459--1467.

\bibitem{JafariSN-2015-HA}
S.~Jafari, J.~Sprott, F.~Nazarimehr, Recent new examples of hidden attractors,
  European Physical Journal: Special Topics 224~(8) (2015) 1469--1476.

\bibitem{ZhusubaliyevMCM-2015-HA}
Z.~Zhusubaliyev, E.~Mosekilde, A.~Churilov, A.~Medvedev, Multistability and
  hidden attractors in an impulsive {G}oodwin oscillator with time delay,
  European Physical Journal: Special Topics 224~(8) (2015) 1519--1539.

\bibitem{SahaSRC-2015-HA}
P.~Saha, D.~Saha, A.~Ray, A.~Chowdhury, Memristive non-linear system and hidden
  attractor, European Physical Journal: Special Topics 224~(8) (2015)
  1563--1574.

\bibitem{Semenov20151553}
V.~Semenov, I.~Korneev, P.~Arinushkin, G.~Strelkova, T.~Vadivasova,
  V.~Anishchenko, Numerical and experimental studies of attractors in
  memristor-based {C}hua's oscillator with a line of equilibria.
  {N}oise-induced effects, European Physical Journal: Special Topics 224~(8)
  (2015) 1553--1561.

\bibitem{FengW-2015-HA}
Y.~Feng, Z.~Wei, Delayed feedback control and bifurcation analysis of the
  generalized {S}prott {B} system with hidden attractors, European Physical
  Journal: Special Topics 224~(8) (2015) 1619--1636.

\bibitem{Li20151493}
C.~Li, W.~Hu, J.~Sprott, X.~Wang, Multistability in symmetric chaotic systems,
  European Physical Journal: Special Topics 224~(8) (2015) 1493--1506.

\bibitem{FengPW-2015-HA}
Y.~Feng, J.~Pu, Z.~Wei, Switched generalized function projective
  synchronization of two hyperchaotic systems with hidden attractors, European
  Physical Journal: Special Topics 224~(8) (2015) 1593--1604.

\bibitem{Sprott20151409}
J.~Sprott, Strange attractors with various equilibrium types, European Physical
  Journal: Special Topics 224~(8) (2015) 1409--1419.

\bibitem{PhamVVJ-2015-HA}
V.~Pham, S.~Vaidyanathan, C.~Volos, S.~Jafari, Hidden attractors in a chaotic
  system with an exponential nonlinear term, European Physical Journal: Special
  Topics 224~(8) (2015) 1507--1517.

\bibitem{VaidyanathanPV-2015-HA}
S.~Vaidyanathan, V.-T. Pham, C.~Volos, A 5-{D} hyperchaotic {R}ikitake dynamo
  system with hidden attractors, European Physical Journal: Special Topics
  224~(8) (2015) 1575--1592.

\bibitem{Danca-2016-HA}
M.-F. Danca, Hidden transient chaotic attractors of {R}abinovich--{F}abrikant
  system, Nonlinear Dynamics 86~(2) (2016) 1263--1270.

\bibitem{Zelinka-2016-HA}
I.~Zelinka, Evolutionary identification of hidden chaotic attractors,
  Engineering Applications of Artificial Intelligence 50 (2016) 159--167.

\bibitem{LeonovKM-2015-CNSNS}
G.~Leonov, N.~Kuznetsov, T.~Mokaev, Hidden attractor and homoclinic orbit in
  {L}orenz-like system describing convective fluid motion in rotating cavity,
  Communications in Nonlinear Science and Numerical Simulation 28 (2015)
  166--174.
\newblock \href {http://dx.doi.org/10.1016/j.cnsns.2015.04.007}
  {\path{doi:10.1016/j.cnsns.2015.04.007}}.

\bibitem{KaplanY-1979}
J.~L. Kaplan, J.~A. Yorke, Chaotic behavior of multidimensional difference
  equations, in: Functional Differential Equations and Approximations of Fixed
  Points, Springer, Berlin, 1979, pp. 204--227.

\bibitem{DouadyO-1980}
A.~Douady, J.~Oesterle, Dimension de {H}ausdorff des attracteurs, C.R. Acad.
  Sci. Paris, Ser. A. (in French) 290~(24) (1980) 1135--1138.

\bibitem{Ledrappier-1981}
F.~Ledrappier, Some relations between dimension and {L}yapounov exponents,
  Communications in Mathematical Physics 81~(2) (1981) 229--238.

\bibitem{EdenFT-1991}
A.~Eden, C.~Foias, R.~Temam, Local and global {L}yapunov exponents, Journal of
  Dynamics and Differential Equations 3~(1) (1991) 133--177, [Preprint No.
  8804, The Institute for Applied Mathematics and Scientific Computing, Indiana
  University, 1988].
\newblock \href {http://dx.doi.org/10.1007/BF01049491}
  {\path{doi:10.1007/BF01049491}}.

\bibitem{RusselHO-1980}
D.~Russel, J.~Hanson, E.~Ott, Dimension of strange attractors, Physical Review
  Letters 45~(14) (1980) 1175--1178.

\bibitem{Leonov-1991-Vest}
G.~A. Leonov, On estimations of {H}ausdorff dimension of attractors, Vestnik
  St. Petersburg University: Mathematics 24~(3) (1991) 38--41, [Transl. from
  Russian: Vestnik Leningradskogo Universiteta. Mathematika, 24(3), 1991,
  pp.~41-44].

\bibitem{BoichenkoLR-2005}
V.~A. Boichenko, G.~A. Leonov, V.~Reitmann, Dimension Theory for Ordinary
  Differential Equations, Teubner, Stuttgart, 2005.

\bibitem{Kuznetsov-2016-PLA}
N.~Kuznetsov, The {L}yapunov dimension and its estimation via the {L}eonov
  method, Physics Letters A 380~(25--26) (2016) 2142--2149.
\newblock \href {http://dx.doi.org/10.1016/j.physleta.2016.04.036}
  {\path{doi:10.1016/j.physleta.2016.04.036}}.

\bibitem{Leonov-2002}
G.~Leonov, Lyapunov dimension formulas for {H}enon and {L}orenz attractors,
  St.Petersburg Mathematical Journal 13~(3) (2002) 453--464.

\bibitem{LeonovP-2005}
G.~Leonov, M.~Poltinnikova, On the {L}yapunov dimension of the attractor of
  {C}hirikov dissipative mapping, AMS Translations. Proceedings of
  St.Petersburg Mathematical Society. Vol. X 224 (2005) 15--28.

\bibitem{Leonov-2012-PMM}
G.~A. Leonov, Lyapunov functions in the attractors dimension theory, Journal of
  Applied Mathematics and Mechanics 76~(2) (2012) 129--141.

\bibitem{LeonovPS-2013-PLA}
G.~Leonov, A.~Pogromsky, K.~Starkov, Erratum to ''{T}he dimension formula for
  the {L}orenz attractor'' [{P}hys. {L}ett. {A} 375 (8) (2011) 1179], Physics
  Letters A 376~(45) (2012) 3472 -- 3474.

\bibitem{LeonovKKK-2015-arXiv-YangTigan}
G.~Leonov, N.~Kuznetsov, N.~Korzhemanova, D.~Kusakin, {L}yapunov dimension
  formula of attractors in the {T}igan and {Y}ang systems, arXiv:1510.01492v1
  http://arxiv.org/pdf/1510.01492v1.pdf.

\bibitem{LeonovAK-2015}
G.~Leonov, T.~Alexeeva, N.~Kuznetsov, Analytic exact upper bound for the
  {L}yapunov dimension of the {S}himizu-{M}orioka system, Entropy 17~(7) (2015)
  5101--5116.
\newblock \href {http://dx.doi.org/10.3390/e17075101}
  {\path{doi:10.3390/e17075101}}.

\bibitem{LeonovKKK-2016-CNSCS}
G.~Leonov, N.~Kuznetsov, N.~Korzhemanova, D.~Kusakin, {L}yapunov dimension
  formula for the global attractor of the {L}orenz system, Communications in
  Nonlinear Science and Numerical Simulation 41 (2016) 84--103.
\newblock \href {http://dx.doi.org/10.1016/j.cnsns.2016.04.032}
  {\path{doi:10.1016/j.cnsns.2016.04.032}}.

\bibitem{HurewiczW-1941}
W.~Hurewicz, H.~Wallman, Dimension Theory, Princeton University Press,
  Princeton, 1941.

\bibitem{OttWY-1984}
E.~Ott, W.~Withers, J.~Yorke, Is the dimension of chaotic attractors invariant
  under coordinate changes?, Journal of Statistical Physics 36~(5-6) (1984)
  687--697.

\bibitem{KuznetsovAL-2016}
N.~Kuznetsov, T.~Alexeeva, G.~Leonov, Invariance of {L}yapunov exponents and
  {L}yapunov dimension for regular and irregular linearizations, Nonlinear
  Dynamics 85~(1) (2016) 195--201.
\newblock \href {http://dx.doi.org/10.1007/s11071-016-2678-4}
  {\path{doi:10.1007/s11071-016-2678-4}}.

\bibitem{LeonovK-2015-AMC}
G.~Leonov, N.~Kuznetsov, On differences and similarities in the analysis of
  {L}orenz, {C}hen, and {L}u systems, Applied Mathematics and Computation 256
  (2015) 334--343.
\newblock \href {http://dx.doi.org/10.1016/j.amc.2014.12.132}
  {\path{doi:10.1016/j.amc.2014.12.132}}.

\bibitem{BenettinGGS-1980-Part2}
G.~Benettin, L.~Galgani, A.~Giorgilli, J.-M. Strelcyn, {L}yapunov
  characteristic exponents for smooth dynamical systems and for hamiltonian
  systems. {A} method for computing all of them. {P}art 2: Numerical
  application, Meccanica 15~(1) (1980) 21--30.

\bibitem{WolfSSV-1985}
A.~Wolf, J.~B. Swift, H.~L. Swinney, J.~A. Vastano, Determining {L}yapunov
  exponents from a time series, Physica D: Nonlinear Phenomena 16~(D) (1985)
  285--317.

\bibitem{Lyapunov-1892}
A.~M. Lyapunov, The General Problem of the Stability of Motion (in Russian),
  Kharkov, 1892, [English transl.: Academic Press, NY, 1966].

\bibitem{TempkinY-2007}
J.~Tempkin, J.~Yorke, Spurious {L}yapunov exponents computed from data, SIAM
  Journal on Applied Dynamical Systems 6~(2) (2007) 457--474.
\newblock \href {http://dx.doi.org/10.1137/040619211}
  {\path{doi:10.1137/040619211}}.

\bibitem{AugustovaBC-2015}
P.~Augustova, Z.~Beran, S.~Celikovsky, {ISCS} 2014: Interdisciplinary Symposium
  on Complex Systems, Emergence, Complexity and Computation (Eds.: {A}.
  {S}anayei et al.), Springer, 2015, Ch. On Some False Chaos Indicators When
  Analyzing Sampled Data, pp. 249--258.

\bibitem{FredericksonKYY-1983}
P.~Frederickson, J.~Kaplan, E.~Yorke, J.~Yorke, The {L}iapunov dimension of
  strange attractors, Journal of Differential Equations 49~(2) (1983) 185--207.

\bibitem{FarmerOY-1983}
J.~Farmer, E.~Ott, J.~Yorke, The dimension of chaotic attractors, Physica D:
  Nonlinear Phenomena 7~(1-3) (1983) 153 -- 180.

\bibitem{BogoliubovK-1937}
N.~Bogoliubov, N.~Krylov, La theorie generalie de la mesure dans son
  application a l'etude de systemes dynamiques de la mecanique non-lineaire,
  Ann. Math. II (in French) (Annals of Mathematics) 38~(1) (1937) 65--113.

\bibitem{DellnitzJ-2002}
M.~Dellnitz, O.~Junge, Set oriented numerical methods for dynamical systems,
  in: Handbook of Dynamical Systems, Vol.~2, Elsevier Science, 2002, pp.
  221--264.

\bibitem{Oseledec-1968}
V.~Oseledets, A multiplicative ergodic theorem. {C}haracteristic {L}japunov,
  exponents of dynamical systems, Trudy Moskovskogo Matematicheskogo
  Obshchestva (in Russian) 19 (1968) 179--210.

\bibitem{BarreiraS-2000}
L.~Barreira, J.~Schmeling, Sets of ``{N}on-typical'' points have full
  topological entropy and full {H}ausdorff dimension, Israel Journal of
  Mathematics 116~(1) (2000) 29--70.

\bibitem{ChaosBook}
P.~Cvitanovi\'c, R.~Artuso, R.~Mainieri, G.~Tanner, G.~Vattay, Chaos: Classical
  and Quantum, Niels Bohr Institute, Copenhagen, 2012, {h}ttp://ChaosBook.org.

\bibitem{OttY-2008}
W.~Ott, J.~Yorke, When {L}yapunov exponents fail to exist, Phys. Rev. E 78
  (2008) 056203.

\bibitem{Young-2013}
L.-S. Young, Mathematical theory of {L}yapunov exponents, Journal of Physics A:
  Mathematical and Theoretical 46~(25) (2013) 254001.

\bibitem{LeonovK-2007}
G.~Leonov, N.~Kuznetsov, Time-varying linearization and the {P}erron effects,
  International Journal of Bifurcation and Chaos 17~(4) (2007) 1079--1107.
\newblock \href {http://dx.doi.org/10.1142/S0218127407017732}
  {\path{doi:10.1142/S0218127407017732}}.

\bibitem{Schmeling-1998}
J.~Schmeling, A dimension formula for endomorphisms -- the {B}elykh family,
  Ergodic Theory and Dynamical Systems 18 (1998) 1283--1309.

\bibitem{KuznetsovMV-2014-CNSNS}
N.~Kuznetsov, T.~Mokaev, P.~Vasilyev, Numerical justification of {L}eonov
  conjecture on {L}yapunov dimension of {R}ossler attractor, Commun Nonlinear
  Sci Numer Simulat 19 (2014) 1027--1034.

\bibitem{Smith-1986}
R.~Smith, Some application of {H}ausdorff dimension inequalities for ordinary
  differential equation, Proc. Royal Society Edinburg 104A (1986) 235--259.

\bibitem{Gelfert-2003}
K.~Gelfert, Maximum local {L}yapunov dimension bounds the box dimension.
  {D}irect proof for invariant sets on {R}iemannian manifolds, Z. Anal. Anwend.
  22 (2003) 553--568.

\bibitem{Eden-1990}
A.~Eden, Local estimates for the {H}ausdorff dimension of an attractor, Journal
  of Mathematical Analysis and Applications 150~(1) (1990) 100--119.

\bibitem{ConstantinFT-1985}
P.~Constantin, C.~Foias, R.~Temam, Attractors representing turbulent flows,
  Memoirs of the American Mathematical Society 53~(314).

\bibitem{KuznetsovL-2016-ArXiv}
N.~Kuznetsov, G.~Leonov, A short survey on {L}yapunov dimension for finite
  dimensional dynamical systems in {E}uclidean space,
  arXivHttp://arxiv.org/pdf/1510.03835v2.pdf.

\bibitem{DoeringGHN-1987}
C.~Doering, J.~Gibbon, D.~Holm, B.~Nicolaenko, Exact {L}yapunov dimension of
  the universal attractor for the complex {G}inzburg-{L}andau equation, Phys.
  Rev. Lett. 59 (1987) 2911--2914.

\bibitem{LeonovM-2016-MAIK}
G.~A. Leonov, T.~N. Mokaev, Lyapunov dimension formula for the attractor of the
  {G}lukhovsky-{D}olzhansky system, Doklady Mathematics 93~(1) (2016) 42--45.

\bibitem{Eden-1989-PhD}
A.~Eden, An abstract theory of {L}-exponents with applications to dimension
  analysis ({PhD} thesis), Indiana University, 1989.

\bibitem{LeonovL-1993}
G.~Leonov, S.~Lyashko, {E}den's hypothesis for a {L}orenz system, Vestnik St.
  Petersburg University: Mathematics 26~(3) (1993) 15--18, [Transl. from
  Russian: Vestnik Sankt-Peterburgskogo Universiteta. Ser 1. Matematika, 26(3),
  14-16].

\bibitem{DoeringG-1995}
C.~R. Doering, J.~Gibbon, On the shape and dimension of the {L}orenz attractor,
  Dynamics and Stability of Systems 10~(3) (1995) 255--268.

\bibitem{Kuratowski-1966}
K.~Kuratowski, Topology, Academic press, New York, 1966.

\bibitem{EckmannR-1985}
J.-P. Eckmann, D.~Ruelle, Ergodic theory of chaos and strange attractors,
  Reviews of Modern Physics 57~(3) (1985) 617--656.

\end{thebibliography}

\end{document}